\documentclass[12pt]{article}

\usepackage{amssymb,amsmath,amsfonts,eurosym,geometry,ulem,graphicx,color,setspace,sectsty,comment,footmisc,natbib,pdflscape,subfigure,subcaption,array,hyperref,booktabs,ragged2e,bbm,}
\usepackage{pdfpages} 
\usepackage{csquotes}

\newtheorem{theorem}{Theorem}
\newtheorem{lemma}{Lemma}

\newtheorem{proposition}{Proposition}
\newenvironment{proof}[1][Proof]{\noindent\textbf{#1.} }{\ \rule{0.5em}{0.5em}}

\newcolumntype{L}[1]{>{\raggedright\let\newline\\arraybackslash\hspace{0pt}}m{#1}}
\newcolumntype{C}[1]{>{\centering\let\newline\\arraybackslash\hspace{0pt}}m{#1}}
\newcolumntype{R}[1]{>{\raggedleft\let\newline\\arraybackslash\hspace{0pt}}m{#1}}

\begin{document}

\begin{titlepage}
\title{Breakdown Analysis for Instrumental Variables with Binary Outcomes}
\author{Pedro Picchetti\thanks{\textbf{Pedro Picchetti:} Instituto de Economía, Pontificia Universidad Católica de Chile.\\ E-mail: pedro.picchetti@uc.cl. I thank Alexandre Poirier, Cristine Pinto, Jonathan Roth, Matthew Masten, Peter Hull and Vitor Possebom. I gratefully acknowledge the financial support from FAPESP 2021/13708-8. All errors are my own.\\ }}
\date{\today}
\maketitle
\begin{abstract}

\noindent This paper studies the partial identification of treatment effects in Instrumental Variables (IV) settings with binary outcomes under violations of independence. I derive the identified sets for the treatment parameters of interest in the setting, as well as breakdown values for conclusions regarding the true treatment effects. I derive $\sqrt{N}$-consistent nonparametric estimators for the bounds of treatment effects and for breakdown values. These results can be used to assess the robustness of empirical conclusions obtained under the assumption that the instrument is independent from potential quantities, which is a pervasive concern in studies that use IV methods with observational data. In the empirical application, I show that the conclusions regarding the effects of family size on female unemployment using same-sex siblings as the instrument are highly sensitive to violations of independence. \\

\vspace{0in}
\vspace{0in}

\bigskip
\end{abstract}
\setcounter{page}{0}
\thispagestyle{empty}
\end{titlepage}
\pagebreak \newpage

\doublespacing

\section{Introduction}

Instrumental variables (IV) techniques are among the most widely used empirical tools in social sciences. In the canonical IV setting, the causal effect of a binary treatment is identified by exploiting variations in a binary instrument in the form of the \cite{wald} estimand. Point identification is achieved if the instrument satisfies a set of assumptions. For instance, the instrumental variable must be independent from potential treatments and potential outcomes.

Although in certain cases the independence assumption is readily justified (experimental studies), it is often unverifiable and must be defended by appealing to context specific knowledge, specially in observational studies. In this paper I study what can be learned about treatment effects in IV settings under weaker versions of the instrument independence assumption.

I focus in the case where the outcome is binary. I derive bounds for the first-stage and reduced form parameters, as well as bounds for the Local Average Treatment Effect (LATE) under a bounded dependence assumption called \textit{c-dependence} \citep{mastenpoirier2018}, which bounds the distance between the probability of receiving the instrument given observed covariates and unobserved potential quantities and the probability of being treated given just the observed covariates.

The first-stage parameter, the share of compliers, is partially identified as function of the difference between the probability of assignment given covariates and the probability of assignment given covariates and potential treatments. The reduced form parameter, the intention-to-treat effect, is partially identified as function of the difference between the probability of assignment given covariates and the probability of assignment given covariates and potential outcomes. The LATE is partially identified as a function of both sensitivity parameters.

I identify breakdown values for the first-stage and reduced form, as well as the breakdown frontier for the LATE. Breakdown values are the violations of the independence assumptions under which a particular conclusion no longer holds. For instance, one could be interested to learn under which violations of independence the conclusion that the treatment effect has a particular sign holds. If a researcher is concerned about the external validity of the study, the breakdown analysis of the first-stage is useful to understand under which violations of independence the share of compliers is above a certain share of the study population.

I propose nonparametric estimators for the bounds of causal effects and breakdown values, and derive their asymptotic properties using convergence results for \textit{Hadamard
directional differentiable} functions \citep{fangsantos}.

The identified sets for the LATE are not sharp. I derive sharp bounds for the LATE under a joint \textit{$c$-dependence} assumption for potential outcomes and potential treatments. The bounds of the set can be used to derive the breakdown point for conclusions regarding the LATE.

Monte Carlo simulations show the desirable finite sample properties of the proposed estimators.

For the empirical application, I revisit \cite{angev}, which studies the effects of family size on female employment using same-sex siblings as the instrument, and estimate the identified sets for the share of compliers, the ITT and the LATE under different relaxations of independence. The estimated breakdown values for the LATE and the ITT are not statistically different from zero, which suggests that the conclusions of the study are highly sensitive to violations of independece.

\textbf{Related Literature:} This paper relates broadly to three strands of the causal inference literature. First, ir is connected to the literature on partial identification and sensitivity analysis in IV settings. While most results on the literature focus on partial identification under violations of the exclusion restriction \citep{conley,wang18,mastenporirer21,cinelli}, this paper focuses solely on violations of independence. In that sense, it is similar to \cite{mastenkline} and \cite{jonashesh}, but differs from the former by allowing two-sided noncompliance and from the latter by choosing a different sensitivity parameter.

Second, this paper relates to the literature on the identification of breakdown values, introduced by \cite{horwitzmanski}. My approach to inference follows closely the one introduced in \cite{mastenpoirier2020}. While most of the work in this literature focuses on missing data settings \citep{klinesantos} and selection on observables \citep{mastenpoirier2020}, this is the first paper studying inference for breakdown values in settings with non-compliance.

Finally, this paper is related to the literature on IV settings with binary outcomes, which dates back to the seminal work of \cite{heckman78}. While most prominent work on this literature focuses on the identification of the average structural functions \citep{vytlacilyildiz,shaikhvytlacil} or partial identification of Average Treatment Effects \citep{MACHADO2019522}, this paper is more closely related to \cite{chesher} as it builds on the LATE framework for identification and sensitivity analysis.

\textbf{Outline of the paper:} The rest of the paper is organized as follows: Section 2 describes the framework and target parameters in the setting. Section 3 provides the partial identification results for the case of binary outcomes and in Section 4 I show the identification of the breakdown values. In Section 5 I perform a numerical illustration of the method. Section 6 introduces the estimators and their asymptotic properties. Section 7 presents the partial identification results for the case of joint \textit{$c$-dependence}. Section 8 presents the Monte Carlo simulations, Section 9 presents the empirical application and Section 10 concludes.

\section{Framework}

Let $Z\in\left\{ 0,1\right\}$ denote a binary variable that indicates whether an individual was assigned to treatment ($Z=1$) or control ($Z=0$). In the setting, non-compliance is allowed, which means that not all individuals assigned to treatment will actually take the treatment and not all individuals assigned to control will remain untreated. Rather than determining treatment status, the assignment represents an encouragement towards treatment.

Let $D\in\left\{ 0,1\right\}$ denote the actual treatment status. Define the potential treatment associated to assignment $z$ as $D(z)$. We observe the treatment status

\begin{equation*}
    D=ZD(1)+(1-Z)D(0)
\end{equation*}

Let $Y\in\left\{ 0,1\right\}$ denote the observed binary outcome. The potential outcome associated to assignment $z$ is defined as $Y(D(z),z)$. At first, I allow potential outcomes depend arbitrarily on treatment and assignment. Observed and potential outcomes are related by

\begin{equation*}
    Y=ZY(D(1),1)+(1-Z)Y(D(0),0)
\end{equation*}

Let $X\in\mathcal{S}(X)$ be a vector of observed covariates and $p_{z|x}=\mathbb{P}\left ( Z=z|X=x \right )$ be the observed propensity score for assignment. I maintain the following assumption regarding the joint distribution of $(D(z),Y(D(z),z),Z,X)$ throughout the paper:

\textbf{Assumption 1:} For each $z,z'\in\left\{ 0,1\right\}$ and $x\in\mathcal{S}(X)$:

\begin{enumerate}
    \item $\mathbb{P}\left ( D(z)=1|Z=z',X=x \right )\in \left ( 0,1 \right )$
    \item $\mathbb{P}\left ( Y(D(z),z)=1|Z=z',X=x \right )\in \left ( 0,1 \right )$
    \item $p_{z|x}>0$
\end{enumerate}

Assumptions 1.1 and 1.2 state that the support of potential quantities does not depend on the assignment. Assumption 1.3 states that all individuals can be assigned to treatment and control with probability greater than zero, and is usually referred to as the common support, or overlap assumption.

The fundamental behavioral assumption for identification in IV settings restricts how individuals respond to assignment, and is formalized below:

\textbf{Assumption 2:} For all $x\in\mathcal{S}(X)$, we have $D(1)\geq D(0)$ conditional on $X=x$.

Assumption 2 is referred to as the monotonicity condition \citep{imbensangrist} and it states that there are no individuals that would take treatment if assigned to control and would not take treatment in the presence of the encouragement. Under assumption 2 individuals can be divided into three groups regarding their response to assignment: Always-takers (individuals that take treatment regardless of the assignment), Compliers (individuals that mimick their assignment) and Never-takers (individuals that don't take treatment regardless of their assignment).

Another necessary assumption for identification is the exclusion restriction:

\textbf{Assumption 3:} For all $x\in\mathcal{S}(X)$ and $z\in\left\{ 0,1\right\}$, $Y(D(z),z)=Y(D(z
))$.

The exclusion restriction states that the outcome is only affected directly by the actual uptake of the treatment. Hence, assignment only affects outcomes to the extent that it affects the choice of treatment. Under the exclusion restriction, the observed outcome relates to potential outcomes simply by $Y=ZY(D(1))+(1-Z)Y(D(0))$.

Point identification in IV settings usually relies on two additional assumptions, which are stated below:

\textbf{Assumption 4:} For all $x\in\mathcal{S}(X)$, $\mathbb{E}\left [ D|Z=1,X=x \right ]\neq\mathbb{E}\left [D|Z=0,X=x  \right ]$

Assumption 4 is a technical assumption often referred to as relevance.

\textbf{Assumption 5:} For all $x\in\mathcal{S}(X)$, $\left (Y(D(z)),D(z)  \right )\perp Z|X=x$.

Assumption 5 states that the assignment of treatment is independent of the potential quantities. Although it is usually justified in experimental settings, it is hard to justify and verify in observational settings.

Under these five assumptions, it is well known that the average treatment effect for compliers (LATE) conditional on $X_{i}=x$ is identified by the conditional Wald estimand:

\begin{align*}
    &\mathbb{E}\left [ Y(1)-Y(0)|D(1)>D(0),X=x \right ]=\frac{\mathbb{E}\left [ Y|Z=1,X=x \right ]-\mathbb{E}\left [ Y|Z=0,X=x \right ]}{\mathbb{E}\left [ D|Z=1,X=x \right ]-\mathbb{E}\left [ D|Z=0,X=x \right ]}\\&=\frac{\mathbb{E}\left [ Y(D(1))-Y(D(0))|X=x \right ]}{\mathbb{E}\left [ D(1)-D(0)|X=x \right ]}\equiv\frac{\tau_{Y}(x)}{\tau_{D}(x)}\equiv\tau(x)
\end{align*}

The unconditional LATE is identified by integrating $\tau_{Y}(x)$ and $\tau_{D}(x)$ over the distribution of covariates. In this paper, I study the partial identification of the LATE in settings where the independence assumption is violated. The approach consists in finding bounds for the first-stage $(\tau_{D}(x))$ and the reduced form $(\tau_{Y}(x))$ as functions of the magnitude of the dependence of treatment assignment on potential quantities.

The partial identification results are used to conduct sensitivity analysis and identifying breakdown frontiers, that is, the boundary between the set of assumptions which lead to a specific conclusion and those which
do not. For instance, one might be interested in the values of dependence which change the conclusion that the LATE is greater than zero.

\section{Partial Identification with Binary Outcomes}

I consider the partial identification as a function of violations of independence in a setting where the outcome $Y$ is binary. In that case, the conditional LATE can be interpreted as the increase in probability of observing $Y=1$ due to the treatment for compliers with covariates $X=x$:

\begin{equation*}
    \tau(x)=\mathbb{P}\left (Y(1)=1 |D(1)>D(0),X=x \right )-\mathbb{P}\left (Y(0)=1 |D(1)>D(0),X=x \right )
\end{equation*}

I begin with the partial identification of the share of compliers. 

\subsection{First-Stage}

I begin with the partial identification of $\tau_{D}(x)$. For that purpose, write $\tau_{D}(x)=\tau_{D(1)}(x)-\tau_{D(0)}(x)$, where $\tau_{D(z)}(x)=\mathbb{E}\left [D(z)|X=x  \right ]$. The parameter $\tau_{D}(x)$ can be interpreted as the share of compliers with $X=x$: $\mathbb{P}\left(D(1)>D(0)|X=x\right)$. In the case where Assumptions 1-5 and independence hold, $\tau_{D}(x)$ is identified by

\begin{equation*}
    \tau_{D}(x)=\mathbb{E}\left [D|Z=1,X=x  \right ]-\mathbb{E}\left [D|Z=0,X=x  \right ]
\end{equation*}

which is usually referred to as the first-stage in the Wald estimand. 

We replace the independence assumption by a bounded dependence assumption, called \textit{c-dependence} \citep{mastenpoirier2018}:

\textbf{Definition:} Let $x\in\mathcal{S}(X)$. Let $c_{1}$ be a scalar between 0 and 1. We say that $Z$ is \textit{$c_1$-dependent} with potential treatment $D(z)$ given $X=x$ if

\begin{equation*}
    \sup_{d\in\left\{ 0,1\right\}}\left | \mathbb{P}\left ( Z=1|D(z)=d,X=x \right )-\mathbb{P}\left ( Z=1|X=x \right )\right |\leq c_{1}
\end{equation*}
\

Under \textit{$c_{1}$-dependence}, the unobserved propensity score is allowed to deviate $c_1$ probability units away from the observed propensity score $p_{1|x}$. For $c_{1}=0$, the assignment is independent from potential treatments (Assumption 5 holds). Throughout the paper, I assume \textit{$c_1$-dependence}:

\textbf{Assumption 5A:} $Z$ is \textit{$c_1$-dependent} with $D(1)$ given $X$ and with $D(0)$ given $X$.

Let $p_{D|z,x}=\mathbb{P}\left ( D=1|Z=z,X=x \right )$. Lemma 1 provides sharp identified sets for potential treatments and the share of compliers:

\begin{lemma}
    Suppose Assumptions 1-3 and 5A hold. Then, the sharp identified set for potential treatment associated to assignment $z$ is $\tau_{D(z)}(x)\in\left[\tau^{LB}_{D(z)}(c_{1},x),\tau^{UB}_{D(z)}(c_{1},x)\right]$ is, where

    \begin{equation*}
        \tau^{LB}_{D(z)}(c_{1},x)=\max\left\{ \frac{p_{D|z,x}p_{z|x}}{p_{z|x}+c_{1}},\frac{p_{D|z,x}p_{z|x}-c_{1}}{p_{z|x}-c_{1}},p_{D|z,x}p_{z|x}\right\}
    \end{equation*}

    and

    \begin{equation*}
        \tau_{D(z)}^{UB}(c_{1},x)=\min\left\{ \frac{p_{D|z,x}p_{z|x}}{p_{z|x}-c_{1}}\mathbf{1}\left ( p_{z|x}>c_{1} \right )+\mathbf{1}\left ( p_{z|x}\leq c_{1} \right ),\frac{p_{D|z,x}p_{z|x}+c_{1}}{p_{z|x}+c_{1}},p_{D|z,x}p_{z|x}+(1-p_{z|x})\right\}
    \end{equation*}

    Consequently, the sharp identified set for the share of compliers is $\tau_{D}(x)\left[\tau^{LB}_{D}(c_{1},x),\tau^{UB}_{D}(c_{1},x)\right]$, where

    \begin{equation*}
        \tau_{D}^{LB}(c_{1},x)=\max\left ( 0,\tau^{LB}_{D(1)}(c_{1},x)-\tau^{UB}_{D(0)}(c_{1},x) \right )
    \end{equation*}

    and

    \begin{equation*}
        \tau_{D}^{UB}(c_{1},x)=\tau^{UB}_{D(1)}(c_{1},x)-\tau_{D(0)}^{LB}(c_{1},x)
    \end{equation*}

\end{lemma}

Lemma 1 follows directly from Proposition 5 of \cite{mastenpoirier2018}. The upper bound for the first-stage is identified by the difference between the upper bound of $\tau_{D(1)}(c_{1},x)$ and the lower bound of $\tau_{D(0)}(c_{1},x)$. These are both quantities between zero and one, and the monotonicity assumptions ensures that the difference between these quantities is positive.

The lower bound is identified by the difference beteween the lower bound of $\tau_{D(1)}(c_{1},x)$ and the upper bound of $\tau_{D(0)}(c_{1},x)$. There is no guarantee that these difference is greater than zero. Since the first-stage identifies a share between zero and one, the lower bound for is restricted to be at least as great as zero.

The unconditional bounds for the first-stage are obtained by integrating the conditional bounds over the distribution of covariates:

\vspace{-10mm}

\begin{align*}
    &\tau_{D}^{LB}(c_{1})=\max\left ( 0,\int\tau^{LB}_{D(1)}(c_{1},x)dF_{X}(x)-\int\tau^{UB}_{D(0)}dF_{X}(x)(c_{1},x) \right )\\&\tau_{D}^{UB}(c_{1})=\int\tau^{UB}_{D(1)}(c_{1},x)dF_{X}(x)-\int\tau_{D(0)}^{LB}(c_{1},x)dF_{X}(x)
\end{align*}

\subsection{Reduced Form}

Now, I focus on the identification of $\tau_{Y}(x)$, which I write as $\tau_{Y}(x)=\tau_{Y(D(1))}(x)-\tau_{Y(D(0))}(x)$, where $\tau_{Y(D(z))}(x)=\mathbb{E}\left [ Y(D(z))|X=x \right ]$. The parameter $\tau_{Y}(x)$ captures the effect of assignment on potential outcomes, which is often referred to in the literature as the Intention-to-Treat effect (ITT). In the case where Assumptions 1-5 and independence hold, $\tau_{Y}(x)$ is identified by

\begin{equation*}
    \tau_{Y}(x)=\mathbb{E}\left [Y|Z=1,X=x  \right ]-\mathbb{E}\left [Y|Z=0,X=x  \right ]
\end{equation*}

which is referred to as the reduced form. As in the first-stage, I relax the independence assumption replace it by a \textit{c-dependence} assumption of $Z_{i}$ with the potential outcomes.

\textbf{Definition:} Let $x\in\mathcal{S}(X)$. Let $c_{2}$ be a scalar between 0 and 1. We say that $Z$ is \textit{$c_2$-dependent} with potential outcome $Y(D(z))$ given $X=x$ if

\begin{equation*}
    \sup_{y\in\left\{ 0,1\right\}}\left | \mathbb{P}\left ( Z=1|Y(D(z))=y,X=x \right )-\mathbb{P}\left ( Z=1|X=x \right )\right |\leq c_{2}
\end{equation*}

For $c_{2}=0$, the assignment is independent from potential outcomes (Assumption 5 holds). Throughout the paper, I assume \textit{$c_2$-dependence}:

\textbf{Assumption 5B:} $Z$ is \textit{$c_2$-dependent} with  $Y(D(1))$ given $X$ and with $Y(D(0))$ given $X$.

Let $p_{Y|z,x}=\mathbb{P}\left ( Y=1|Z=z,X=x \right )$. I use the results from \cite{mastenpoirier2018} to derive sharp identified sets for potential outcomes and the ITT:

\begin{lemma}
    Suppose Assumptions 1-3 and 5B hold. Then, the sharp identified set for potential outcome associated to assignment $z$ is $\tau_{Y(D(z))}(x)\in\left[\tau^{LB}_{Y(D(z))}(c_{2},x),\tau^{UB}_{Y(D(z))}(c_{2},x)\right]$ is, where

    \begin{equation*}
        \tau^{LB}_{Y(D(z))}(c_{2},x)=\max\left\{ \frac{p_{Y|z,x}p_{z|x}}{p_{z|x}+c_{2}},\frac{p_{Y|z,x}p_{z|x}-c_{2}}{p_{z|x}-c_{2}},p_{Y|z,x}p_{z|x}\right\}
    \end{equation*}

    and

    \begin{equation*}
        \tau_{Y(D(z))}^{UB}(c_{2},x)=\min\left\{ \frac{p_{Y|z,x}p_{z|x}}{p_{z|x}-c_{2}}\mathbf{1}\left ( p_{z|x}>c_{2} \right )+\mathbf{1}\left ( p_{z|x}\leq c_{2} \right ),\frac{p_{Y|z,x}p_{z|x}+c_{2}}{p_{z|x}+c_{2}},p_{Y|z,x}p_{z|x}+(1-p_{z|x})\right\}
    \end{equation*}

    Consequently, the sharp identified set for the share of compliers is $\tau_{Y}(x)\left[\tau^{LB}_{Y}(c_{2},x),\tau^{UB}_{Y}(c_{2},x)\right]$, where

    \begin{equation*}
        \tau_{Y}^{LB}(c_{2},x)=\max\left ( 0,\tau^{LB}_{Y(D(1))}(c_{2},x)-\tau^{UB}_{Y(D(0))}(c_{2},x) \right )
    \end{equation*}

    and

    \begin{equation*}
        \tau_{Y}^{UB}(c_{2},x)=\tau^{UB}_{Y(D(1))}(c_{2},x)-\tau_{Y(D(0))}^{LB}(c_{2},x)
    \end{equation*}
\end{lemma}

The bounds are similar to those obtained for the first-stage. Note that the ITT is not bounded by definition between zero and one, and thus, there is no need to constraint the lower bound to be at least as great as zero.

The unconditional bounds for the ITT are identified by integrating the conditional bounds over the distribution of covariates:

\vspace{-5mm}

\begin{align*}
    &\tau_{Y}^{LB}(c_{2})=\int\tau^{LB}_{Y(D(1))}(c_{2},x)dF_{X}(x)-\int\tau^{UB}_{Y(D(0))}(c_{2},x)dF_{X}(x)\\&\tau_{Y}^{UB}(c_{2})=\int\tau^{UB}_{Y(D(1))}(c_{2},x)dF_{X}(x)-\int\tau^{LB}_{Y(D(0))}(c_{2},x)dF_{X}(x)
\end{align*}

\subsection{Local Average Treatment Effect}

The Local Average Treatment Effect (LATE), the average treatment effect for the subgroup of compliers, under the standard IV assumptions, is point-identified by the ratio of the reduced form and the first-stage. Replacing Assumption 5 with Assumptions 5A and 5B, we obtain the following bounds for the LATE. Putting the pieces from Sections 3.1 and 3.2 together, we find that $\tau(x)\in\left[\tau^{LB}(c_{1},c_{2},x),\tau^{UB}(c_{1},c_{2},x)\right]$, with

\vspace{-5mm}

\begin{align*}
    &\tau^{LB}(c_{1},c_{2},x)=\max\left(\frac{\tau_{Y}^{LB}(c_{2},x)}{\tau_{D}^{UB}(c_{1},x)},-1\right)\\&\tau^{UB}(c_{1},c_{2},x)=\min\left(\frac{\tau_{Y}^{UB}(c_{2},x)}{\tau_{D}^{LB}(c_{1},x)},1\right)
\end{align*}

In the case of a binary outcome, treatment effects are not greater than $1$ nor smaller than $-1$. Hence, the lower bound is the greatest value between the ratio of the lower bound of the ITT and the upper bound of the first-stage, and $-1$. The upper bound is the smallest value between the ratio of the upper bound of the ITT and the lower bound of the first-stage, and $1$. The unconditional bounds for the LATE are identified by integrating the conditional LATE bounds over the distribution of covariates:

\begin{equation*}
    \tau^{LB}(c_{1},c_{2})=\max\left ( \frac{\int\tau^{LB}_{Y}(c_{2},x)dF_{X}(x)}{\int\tau^{UB}_{D}(c_{1},x)dF_{X}(x)},-1 \right )
\end{equation*}

and

\begin{equation*}
    \tau^{UB}(c_{1},c_{2})=\min\left ( \frac{\int\tau^{UB}_{Y}(c_{2},x)dF_{X}(x)}{\int\tau^{LB}_{D}(c_{1},x)dF_{X}(x)},1 \right )
\end{equation*}

The bounds for the LATE are functions of violations of instrument independence with respect to both potential treatments and potential outcomes. In that sense, this result is similar to the partial identification result presented in Section 5.3 of \cite{jonashesh}, which partially identifies the LATE as a function of the finite-population covariance between the assignment probabilities and the potential outcomes and treatments. In the next section, I show how researchers can identify the set of violations of independence under which a conclusion is invalidated.

\section{Breakdown Analysis}

The fundamental interest here is to understand under which violations of independence a particular conclusion still hold. In this section I define breakdown points for conclusions regarding the firs-stage and the reduced form separately, and a breakdown frontier for conclusions regarding the LATE.

I begin by defining the breakdown point for conclusions of the first-stage. That is, what is the largest value of $c_1$ under which we can conclude that $\mathbb{P}\left(D(1)>D(0)\right)\geq\mu$? First, define the robust region for the conclusion as the set of values of $c_{1}$ under which the conclusion holds:

\begin{equation*}
    RR_{FS}(\mu)=\left\{ c_{1}\in \left [ 0,1 \right ]:\tau^{LB}_{D}(c_{1})\geq\mu\right\}
\end{equation*}

The robust region for the first-stage is the set of values of $c_{1}$ for which the identified set for the share of compliers is above $\mu$. The breakdown point is the value of $c_{1}$ on the boundary of the robust region. Formally, the breakdown point is $c_{1}^{*}$ such that $\tau^{LB}_{D}(c_{1}^{*})=\mu$. The breakdown point is identified by solving the expression $\int\tau^{LB}_{D(1)}(c_{1},x)dF_{X}(x)-\int\tau^{UB}_{D(0)}(c_{1},x)dF_{X}(x)=\mu$ for $c_{1}$. Let $bp_{FS}(\mu)$ denote the solutions of the expression above. The analytical expression for the breakdown point is

\begin{equation*}
    c_{1}^{*}=\min\left\{ \max\left\{ bp_{FS}(\mu),0\right\},1\right\}
\end{equation*} 

That is, the breakdown point is the smallest value of $c_{1}$ which solves the breakdown equation, as long as it is bounded in the unit interval. If that is not case, then the breakdown point depends on the worst-case bounds. If the worst-case bounds lie within the robust region, then the breakdown point is 1. Otherwise, it is equal to 0.

Similarly, when it comes to the reduced form, the breakdown point is the largest value of $c_{2}$ under which one can conclude that $\mathbb{E}\left[Y(D(1))-Y(D(0))\right]\geq\mu$. It is defined implicitly by $c_{2}^{*}$ such that $\tau^{LB}_{Y}(c_{2}^{*})=\mu$. Analogously to the first-stage, the breakdown point for the ITT is identified by solving the expression $\int\tau^{LB}_{Y(D(1))}(c_{2},x)dF_{X}(x)-\int\tau^{UB}_{Y(D(0))}(c_{2},x)dF_{X}(x)=\mu$ for $c_{2}$. Let $bp_{RF}(\mu)$ denote the solutions to the expression above, we obtain the following expression for the breakdown point:

\begin{equation*}
    c_{2}^{*}=\min\left\{ \max\left\{ bf_{RF}(\mu),0\right\},1\right\}
\end{equation*}

When it comes to the LATE, the parameter is partially identified as function of both sensitivity parameters $c_{1},c_{2}$. In that case, the robust region of identification is the set of values of $c_{1}$ and $c_{2}$ under which the conclusion holds. It is defined as

\begin{equation*}
    RR(\mu)=\left\{ \left ( c_{1},c_{2} \right )\in \left [ 0,1 \right ]^{2}:\tau^{LB}(c_{1},c_{2})\geq\mu\right\}
\end{equation*}

The breakdown frontier is the set of values of $c_{1}$ and $c_{2}$ on the boundary of the robust region. Specifically, the breakdown frontier is defined as 

\begin{equation*}
    BF(\mu)=\left\{ \left ( c_{1},c_{2} \right )\in \left [ 0,1 \right ]^{2}:\tau^{LB}(c_{1},c_{2})=\mu\right\}
\end{equation*}

Consider the case where $\mu\neq0$. The breakdown frontier can be identified as a function of $c_{1}$ by solving the following equation for $c_{2}$:

\begin{equation*}
    \frac{\int\tau^{LB}_{Y(D(1))}(c_{2},x)dF_{X}(x)-\int\tau^{UB}_{Y(D(0))}(c_{2},x)dF_{X}(x)}{\int\tau^{UB}_{D}(c_{1},x)}=\mu
\end{equation*}

Let $bf(c_{1},\mu)$ denote the solutions, we obtain the following expression for the breakdown frontier as a function of $c_{1}$:

\begin{equation*}
    BF(c_{1},\mu)=\min\left\{ \max\left\{ bf(c_{1},\mu),0\right\},1\right\}
\end{equation*}

Note that in the case where $\mu=0$, the breakdown frontier collapses to the breakdown point for the ITT. That is, $BF(c_{1},0)=bp_{RF}(0)$.

The shape of the breakdown frontier provides insights on the tradeoffs between these two types of relaxations of independence when researchers are drawing conclusions. The relaxations $c_1$ and $c_2$ are measured in the same unit, which helps the interpretation of the breakdown analysis.

\section{Numerical Illustration}

I illustrate the breakdown approach with a simple numerical illustration. Let $X$ be a binary covariate that follows a Bernoulli distribution with parameter $p_{x}=0.5$ Let $Z$ denote the instrument which also follows a Bernoulli distribution with parameter $p_{z}$. Let $p_{z|x}=0.5$, for the sake of simplicity. Selection into treatment follows a threshold-crossing model as in \cite{vytla}:

\begin{equation*}
    D=\mathbf{1}\left\{ \pi_{z}Z+\pi_{x}X\geq V\right\}
\end{equation*}

where $V$ has a standard uniform distribution. The parameter $\pi_{z}$ is the share of compliers, and is set to be equal to 0.5.

The binary outcome also follows a threshold model, as in Heckman (1978):

\begin{equation*}
    Y=\mathbf{1}\left\{ \beta_{d}D+\beta_{x}X\geq U\right\}
\end{equation*}

where $U$ is also uniformly distributed. The random variables $U$ and $V$ are linearly correlated as in \cite{olsen}, which generates the selection problem. The parameter $\beta_{d}$ is the LATE in this DGP, which is set to be equal to 0.5. Hence, it follows that the ITT is equal to 0.25.

\begin{figure}[t!]
\justify     
\caption{Identified sets for the first-stage and reduced form.}
\label{fig:terms}
\subfigure[First-stage]{\label{fig:terms_overtime}\includegraphics[width=80mm]{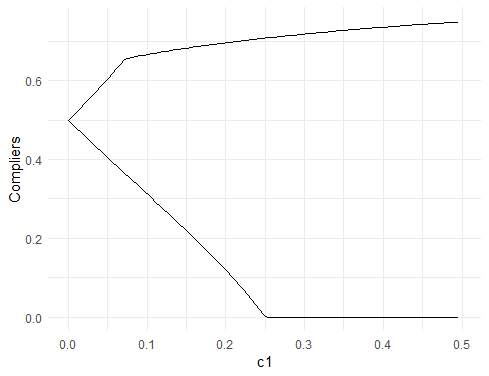}}
\subfigure[Reduced form]{\label{fig:map}\includegraphics[width=80mm]{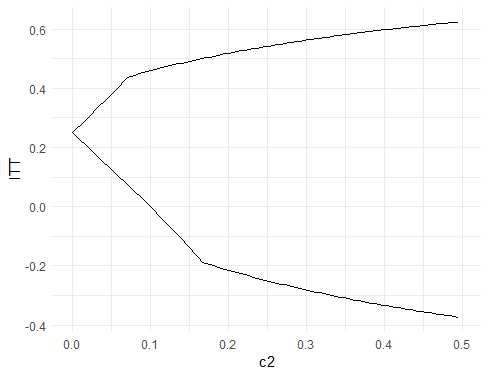}}\\

\scriptsize \noindent \textit{Note:} Figure (a) show the identified set for the share of compliers under the dgp described in Section 6 and different values of $c_{1}$. Figure (b) show the identified set for the ITT under the dgp described in Section 6 and different values of $c_{2}$
\end{figure}

Figure 1 (a) shows the identified set for the first-stage as a function of $c_{1}$. The breakdown point for the conclusion that the share of compliers is greater than zero is $c_{1}^{*}=0.25$. Figure 1 (b) show the identified set of the ITT as a function of $c_{2}$. The breakdown point for the conclusion that the ITT, and hence the LATE, is greater than zero is $c_{1}^{
*}=0.1$. 

Figure 2 shows the breakdown frontier for the LATE. I specify the breakdown frontier for the conclusion that the LATE is greater than 0.25, which is half of its true value. The blue area represents the robust region, that is, this is the set of values for violations of independence under which the conclusion still holds.

\begin{figure}[t]
    \justify
    \centering
    \caption{Breakdown Frontier for the LATE}
    \includegraphics{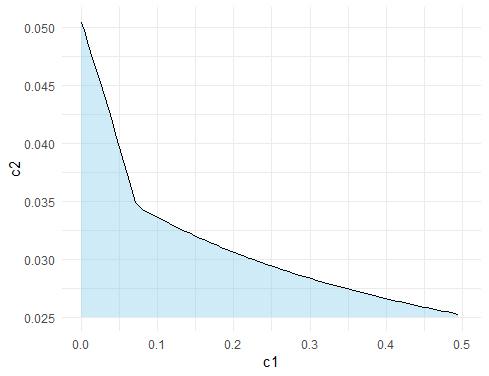}
    \label{fig:enter-label}\\
    \scriptsize \noindent \textit{Note:} Figure 2 plots the breakdown frontier for the conclusion that $\tau\geq0.25$.
\end{figure}

\section{Estimation and Inference}

In this section I study estimation and inference on the identified sets and breakdown values defined above. The estimands for the bounds of assignment effects, the LATE and the breakdown values are  functionals of the conditional cdf of outcomes and treatment given assignment and covariates, the probability of assignment given covariates, and the marginal distribution of the covariates. I propose nonparametric sample analog estimators and derive asymptotic distributional results using a delta method for directionally differentiable functionals. First, I assume we observe a random sample of data.

\textbf{Assumption 6:} The random variables $\left\{ \left ( Y_{i},D_{i},Z_{i},X_{i} \right )\right\}_{i=1}^{n}$ are independently and identically distributed according to the distribution of $\left(Y,D,Z,X\right)$.

Furthermore, I assume the support of covariates is discrete.

\textbf{Assumption 7:} The support of X is discrete and finite. Let $\mathcal{S}(X)=\left\{ x_{1},...,x_{K}\right\}$ up to a finite $K$.

All target parameters are functionals of the underlying parameters $p_{Y|z,x}$, $p_{D|z,x}$, $p_{z|x}$ and $q_{x}=\mathbb{P}\left(X=x\right)$. Their parametric estimators are, respectively, 

\begin{equation*}
    \widehat{p}_{Y|z,x}=\frac{\frac{1}{N}\sum_{i=1}^{N}\mathbf{1}\left ( Y_{i}=1 \right )\mathbf{1}\left ( Z_{i}=z,X_{i}=x \right )}{\frac{1}{N}\sum_{i=1}^{N}\mathbf{1}\left ( Z_{i}=z,X_{i}=x \right )},
\end{equation*}

\begin{equation*}
    \widehat{p}_{D|z,x}=\frac{\frac{1}{N}\sum_{i=1}^{N}\mathbf{1}\left ( D_{i}=1 \right )\mathbf{1}\left ( Z_{i}=z,X_{i}=x \right )}{\frac{1}{N}\sum_{i=1}^{N}\mathbf{1}\left ( Z_{i}=z,X_{i}=x \right )},
\end{equation*}

\begin{equation*}
    \widehat{p}_{z|x}=\frac{\frac{1}{N}\sum_{i=1}^{N}\mathbf{1}\left ( Z_{i}=z,X_{i}=x \right )}{\frac{1}{N}\sum_{i=1}^{N}\mathbf{1}\left ( X_{i}=x \right )}
\end{equation*}and

\begin{equation*}
    \widehat{q}_{x}=\frac{1}{N}\sum_{i=1}^{N}\mathbf{1}\left ( X_{i}=x \right )
\end{equation*}

These quantities converge uniformly to a Gaussian process at a $\sqrt{N}$-rate; see Lemma B1 in Appendix B. Next, consider the bounds for potential treatments and outcomes. I estimate these bounds by a plug-in estimator of the quantities above. First, I introduce an additional assumption:

\textbf{Assumption 8:} For all $x\in\mathcal{S}(X)$, we have $\max\left\{ c_{1},c_{2}\right\}<\min\left\{ p_{1|x},p_{0|x}\right\}$ and $\tau^{LB}_{D}(c_{1},x)>0$.

Assumption 8 is a technical assumption which simplifies the bounds for potential treatments to

\vspace{-10mm}

\begin{align*}
    &\tau^{LB}_{D(z)}(c_{1},x)=\max\left\{ \frac{p_{D|z,x}p_{z|x}}{p_{z|x}+c_{1}},\frac{p_{D|z,x}p_{z|x}-c_{1}}{p_{z|x}-c_{1}},p_{D|z,x}p_{z|x}\right\}\\&\tau_{D(z)}^{UB}(c_{1},x)=\min\left\{ \frac{p_{D|z,x}p_{z|x}}{p_{z|x}-c_{1}},\frac{p_{D|z,x}p_{z|x}+c_{1}}{p_{z|x}+c_{1}},p_{D|z,x}p_{z|x}+(1-p_{z|x})\right\}
\end{align*}

and the bounds for potential outcomes to

\vspace{-10mm}

\begin{align*}
    &\tau^{LB}_{Y(D(z))}(c_{2},x)=\max\left\{ \frac{p_{Y|z,x}p_{z|x}}{p_{z|x}+c_{2}},\frac{p_{Y|z,x}p_{z|x}-c_{2}}{p_{z|x}-c_{2}},p_{Y|z,x}p_{z|x}\right\}\\&\tau_{Y(D(z))}^{UB}(c_{2},x)=\min\left\{ \frac{p_{Y|z,x}p_{z|x}}{p_{z|x}-c_{2}},\frac{p_{Y|z,x}p_{z|x}+c_{2}}{p_{z|x}+c_{2}},p_{Y|z,x}p_{z|x}+(1-p_{z|x})\right\}
\end{align*}

This simplification is particularly important to guarantee \textit{Hadamard directional differentiablity} of the estimators. Also, it modifies that standard relevance assumption to the partially identified case, ensuring that the upper bound for the LATE is not subject to the weak instrument problem.

The bounds for potential quantities are estimated by replacing the population quantities in the expressions above by its sample analogues. In Lemma B2 of Appendix B I show that the estimators for potential quantities converge in distribution to a nonstandard distribution given by a continuous transformation of Gaussian processes. This result is the building block for deriving the asymptotic properties of the estimators for bounds and breakdown values. 

First, consider the bounds for the first-stage. The plug-in estimators for the lower and the upper bound, respectively are

\begin{align*}
    &\widehat{\tau}^{LB}_{D}(c_{1},x)=\widehat{\tau}^{LB}_{D(1)}(c_{1},x)-\widehat{\tau}^{UB}_{D(0)}(c_{1},x )\\&\widehat{\tau}^{UB}_{D}(c_{1},x)=\widehat{\tau}^{UB}_{D(1)}(c_{1},x)-\widehat{\tau}^{LB}_{D(0)}(c_{1},x )
\end{align*}

The unconditional bounds are estimated by integrating the estimates of conditional bounds over the empirical distribution of covariates:

\begin{equation*}
    \widehat{\tau}^{LB}_{D}(c_{1})=\frac{1}{N}\sum_{i=1}^{N}\widehat{\tau}^{LB}_{D}(c_{1},X_{i})
\end{equation*}

and

\begin{equation*}
    \widehat{\tau}^{UB}_{D}(c_{1})=\frac{1}{N}\sum_{i=1}^{N}\widehat{\tau}^{UB}_{LB}(c_{1},X_{i})
\end{equation*}

Now consider the bounds for the reduced form. The plug-in estimators are

\begin{align*}
    & \widehat{\tau}_{Y}^{LB}(c_{2},x)=\widehat{\tau}^{LB}_{Y(D(1))}(c_{2},x)-\widehat{\tau}^{UB}_{Y(D(0))}(c_{2},x)\\&\widehat{\tau}_{Y}^{UB}(c_{2},x)=\widehat{\tau}^{UB}_{Y(D(1))}(c_{2},x)-\widehat{\tau}^{LB}_{Y(D(0))}(c_{2},x)
\end{align*}

The unconditional bounds are obtained by integrating the estimates over the empirical distribution of covariates:

\begin{equation*}
    \widehat{\tau}^{LB}_{Y}(c_{2})=\frac{1}{N}\sum_{i=1}^{N}\widehat{\tau}^{LB}_{Y}(c_{2},X_{i})
\end{equation*}

and

\begin{equation*}
    \widehat{\tau}^{UB}_{Y}(c_{2})=\frac{1}{N}\sum_{i=1}^{N}\widehat{\tau}^{UB}_{Y}(c_{2},X_{i})
\end{equation*}

The asymptotic distribution of the estimator for these bounds is formalized in the proposition below:

\begin{proposition}
    Assume Assumptions 1-3, 5A, 5B and 6-8 hold. Then,

    \begin{equation*}
        \sqrt{N}\begin{pmatrix}
\widehat{\tau}_{D}^{LB}(c_{1})-\tau^{LB}_{D}(c_{1}) \\\widehat{\tau}_{D}^{UB}(c_{1})-\tau_{D}^{UB}(c_{1})
\end{pmatrix}\overset{d}{\rightarrow}\textbf{Z}_{FS}(d,z,x,c_{1})
    \end{equation*}

    and

    \begin{equation*}
        \sqrt{N}\begin{pmatrix}
\widehat{\tau}_{Y}^{LB}(c_{2})-\tau^{LB}_{Y}(c_{2}) \\\widehat{\tau}_{Y}^{UB}(c_{2})-\tau_{Y}^{UB}(c_{2})
\end{pmatrix}\overset{d}{\rightarrow}\textbf{Z}_{RF}(y,z,x,c_{2})
    \end{equation*}

    where $\textbf{Z}_{FS}(d,z,x,c_{1})$ and $\textbf{Z}_{RF}(y,z,x,c_{2})$ are Gaussian Elements defined in the Section 1 of Appendix A.
\end{proposition}

Now consider the estimation for the breakdown point for the claim that $\tau_{D}\geq \mu$. We focus on that case where $\tau^{LB}_{D}(0)>\mu$, which implies that $c_{1}^{*}>0$. Define the estimator for the breakdown point $c_{1}^{*}$ as

\begin{equation*}
    \widehat{c}_{1}^{*}=\inf\left\{ c_{1}\in \left [ 0,1 \right ]:\widehat{\tau}_{D}^{LB}(c_{1})\leq\mu\right\}
\end{equation*}

and it is obtained by replacing the population quantities from $bf_{FS}(\mu)$ with the sample analogues defined in this section. The estimator for the breakdown frontier for the first-stage is thus $\widehat{c}_{1}^{*}=\min\left\{ \max\left\{ \widehat{b}f_{FS}(\mu),0\right\},1\right\}$

Similarly, when it comes to the estimation for the breakdown point for the claim that $\tau_{Y}\geq \mu$, define the estimator for the breakdown point $c_{2}^{*}$ as

\begin{equation*}
    \widehat{c}_{2}^{*}=\inf\left\{ c_{2}\in \left [ 0,1 \right ]:\widehat{\tau}_{Y}^{LB}(c_{2})\leq\mu\right\}
\end{equation*}

which is obtained  by replacing the population quantities from $bf_{RF}(\mu)$ with the sample analogues. The estimator for the breakdown frontier for the first-stage is thus $\widehat{c}^{*}_{2}=\min\left\{ \max\left\{ \widehat{b}f_{FS}(\mu),0\right\},1\right\}$.

I now provide a formal result regarding the asymptotic distribution of $\widehat{c}_{1}^{*}$ and $\widehat{c}_{2}^{*}$:

\begin{theorem}
    Assume Assumptions 1-3, 5A, 5B and 6-8 hold. Furthermore, assume that $c_{1},c_{2}\in\left [ 0,\overline{C} \right ]$. Then, 

    \begin{equation*}
        \sqrt{N}\left ( \widehat{c}_{1}^{*}-c_{1}^{*} \right )\overset{d}{\rightarrow}\textbf{Z}_{FS}^{BP}
    \end{equation*}

    and 

    \begin{equation*}
        \sqrt{N}\left ( \widehat{c}_{2}^{*}-c_{2}^{*} \right )\overset{d}{\rightarrow}\textbf{Z}_{RF}^{BP}
    \end{equation*}

    where $\textbf{Z}_{FS}^{BP}$ and $\textbf{Z}_{RF}^{BP}$ are Gaussian random variables defined in Section 2 of Appendix A.
\end{theorem}

Finally, consider the estimators for the bounds and breakdown frontier of the LATE. The estimators for the lower and upper bound are respectively

\begin{equation*}
    \widehat{\tau}^{LB}(c_{1},c_{2})=\max\left\{ \frac{\widehat{\tau}^{LB}_{Y}(c_{2})}{\widehat{\tau}_{D}^{UB}(c_{1})},-1\right\}
\end{equation*}

and

\begin{equation*}
    \widehat{\tau}^{LB}(c_{1},c_{2})=\min\left\{ \frac{\widehat{\tau}^{UB}_{Y}(c_{2})}{\widehat{\tau}_{D}^{LB}(c_{1})},1\right\}
\end{equation*}

The next lemma formalizes the asymptotic distribution of the bounds:

\begin{proposition}
    Suppose Assumptions 1-3, 5A, 5B and 6-8 hold. Then,

    \begin{equation*}
        \sqrt{N}\begin{pmatrix}
\widehat{\tau}^{LB}(c_{1},c_{2})-\tau^{LB}(c_{1},c_{2}) \\
\widehat{\tau}^{UB}(c_{1},c_{2})-\tau^{UB}(c_{1},c_{2})
\end{pmatrix}\overset{d}{\rightarrow}\textbf{Z}_{\tau}(y,d,z,x,c_{1},c_{2})
    \end{equation*}
\end{proposition}

Denote the estimated breakdown frontier for the conclusion that $\tau\geq\mu$ by

\begin{equation*}
    \widehat{BF}(c_{1},\mu)=\min\left\{ \max\left\{ \widehat{bf}(c_{1},\mu),0\right\},1\right\}
\end{equation*}

\

I show that the estimated breakdown frontier converges in distribution.

\begin{theorem}
    Let Assumptions 1-3, 5A, 5B and 6-8 hold. Furthermore, let $\mathcal{C}\subset \left [ 0,\overline{C} \right ]$ and $\mathcal{M}\subset \left [ -1,1 \right ]$ be finite grids of points. Then,

    \begin{equation*}
        \sqrt{N}\left ( \widehat{BF}(c_{1},\mu) -BF(c_{1},\mu)\right )\overset{d}{\rightarrow}\textbf{Z}_{BF}(c_{1},\mu),
    \end{equation*}

a tight random element of $l^{\infty}\left ( \mathcal{C}\times\mathcal{M} \right )$.
\end{theorem}

Since the limiting process is non-Gaussian, inference on the breakdown values will not be based on standard errors. The processes’ distribution is characterized fully by the expressions in Appendices B1 and B2, but obtaining analytical estimates of functionals of these processes is challenging. In the next subsection I give describe a bootstrap procedure that can be used to construct confidence intervals for the breakdown points and confidence bands for the breakdown frontier.

The breakdown points and frontier estimators can be obtained using standard root
finding algorithms, such as Matlab's \textit{fzero} and R's \textit{findZeros}. The solutions provide the estimates.

\subsection{Bootstrap Inference}

Before describing the procedure, I introduce some notation. Let $\mathcal{F}_{i}=\left(Y_{i},D_{i},Z_{i},X_{i}\right)$ and $F_{1:N}=\left\{\mathcal{F}_{1},...,\mathcal{F}_{N}\right\}$. Let $\widehat{\theta}$ denote the estimator of a parameter $\theta_{0}$ based on $\mathcal{F}_{1:N}$. Let $\textbf{A}^{*}_{1:N}\equiv\sqrt{N}\left(\widehat{\theta}^{*}-\widehat{\theta}\right)$ where $\widehat{\theta}^{*}$ is a draw from the nonparametric bootstrap distribution of $\widehat{\theta}$. Suppose $\textbf{A}$ is the tight limiting process of $\sqrt{N}\left(\widehat{\theta}-\theta_{0}\right)$. Bootstrap consistency is given by weak convergence in probability conditional on $\mathcal{F}_{1:N}$, that is,

\begin{equation*}
    \sup_{h\in BL_{1}}\left | \mathbb{E}\left [ h(\textbf{A}^{*}_{1:N})|\mathcal{F}_{1:N} \right ]-\mathbb{E}\left [ h(\textbf{A}) \right ]\right |=o_{p}(1)
\end{equation*}

where $BL_1$ denotes the set of Lipschitz functions into $\mathbb{R}$ with Lipschitz constant no greater than 1. For $\theta_{0}$ and $\widehat{\theta}$, I focus on the parameters introduced in Section 3 and their sample analogue estimators, which are plugged-in in the bounds estimators.

Let $\textbf{Z}_{1:N}=\sqrt{N}\left(\widehat{\theta}^{*}-\widehat{\theta}\right)$. Theorem 3.6.1 of van der Vaart and Wellner (1996) implies the bootstrap consistency of $\textbf{Z}_{1:N}$. Since the parameters of interest are Hadamard differentiable functionals of $\theta_{0}$, it follows from \cite{fangsantos} that the nonparametric bootstrap can be used to do inference on the identified sets and breakdown values.

For the breakdown points of the first-stage and the reduced form, the bootstrap procedure can be used to construct one-sided confidence intervals as in \cite{klinesantos}. For the breakdown frontier of the LATE, the bootstrap can be used to construct uniform one-sided confidence bands as in \cite{mastenpoirier2020}.

\section{Partial Identification under joint \textit{$c$-dependence}}

The results from Sections 3 and 4 provide the breakdown analysis framework for IV settings with binary outcomes in the case where the assumption of instrument independence is replaced by a bounded dependence assumption, that consider the probability of assignment conditional on potential treatments and potential outcomes separately.

Relaxing the independence assumption in terms of $c_1$ and \textit{$c_{2}$-dependence} allows the researcher to conduct breakdown analysis for the share of compliers and the ITT separately, while also allowing for the possibility of assessing tradeoffs between these assumptions in the breakdown frontier for the LATE.

Despite the several desirable features of this approach, it does not provide a sharp identified set for the LATE. In this section, I derive sharp bounds for the LATE under a joint \textit{$c$-dependence assumption}, which I define below:

\textbf{Definition}:  Let $x\in\mathcal{S}(X)$. Let $c$ be a scalar between 0 and 1. We say that $Z$ is joint \textit{$c$-dependent} with $\left(Y(D(z)),D(z)\right)$ given $X=x$ if 

\begin{equation*}
    \sup_{(y,d)\in\left\{0,1\right\}^{2}}\left |\mathbb{P}\left ( Z=1|Y(D(z))=y,D(z)=d,X=x \right )-\mathbb{P}\left ( Z=1|X=x \right ) \right |\leq c
\end{equation*}

As in Section 3, the sensitivity parameter captures deviations from the independence assumption in terms of the distance in probability units between the probability of assignment given covariates and the probability of assignment given covariates and potential quantities. Note that $c=0$ is the case where independence holds, and the target parameters in the setting are point identified. Throughout this section, I assume joint \textit{$c$-dependence}:

\textbf{Assumption 9:} $Z$ is joint \textit{$c$-dependent} with $\left(Y(D(1)),D(1)\right)$ given $X$ and $\left(Y(D(0)),D(0)\right)$ given $X$. Next, I derive the sharp identified set for potential quantities and the LATE:

\begin{theorem}
    For a random variable $Q\in\left\{Y,D\right\}$, denote its potential value associated to assignment $z$ as $Q(z)$. Suppose Assumptions 1-3 and 6-9 hold. The sharp identified set for potential quantities is $\tau_{Q(z)}(x)\in\left [ \tau^{LB}_{Q(z)}(c,x),\tau^{UB}_{Q(z)}(c,x) \right ]$, where

    \begin{align*}
       \tau^{LB}_{Q(z)}(c,x)=\max \left\{ \frac{p_{Q|z,x}p_{z|x}}{p_{z|x}+c},\frac{p_{Q|z,x}p_{z|x}-2c}{p_{z|x}-c},p_{Q|z,x}p_{z|x}\right\}
    \end{align*}

    and

    \begin{align*}
        \tau^{UB}_{Q(z)}(c,x)=\min \left\{ \frac{p_{Q|z,x}p_{z|x}}{p_{z|x}-c},\frac{p_{Q|z,x}p_{z|x}+2c}{p_{z|x}+c},p_{Q|z,x}p_{z|x}+(1-p_{z|x})\right\}
    \end{align*}

    Consequently, the sharp identified set for the LATE is $\tau(x)\in\left [  \tau^{LB}(c,x),\tau^{UB}(c,x)\right ]$, where

    \begin{equation*}
        \tau^{LB}(c,x)=\frac{\tau^{LB}_{Y(D(1))}(c,x)-\tau^{UB}_{Y(D(0))}(c,x)}{\tau^{UB}_{D(1)}(c,x)-\tau^{LB}_{D(0)}(c,x)}
    \end{equation*}

    and

\begin{equation*}
        \tau^{UB}(c,x)=\frac{\tau^{UB}_{Y(D(1))}(c,x)-\tau^{LB}_{Y(D(0))}(c,x)}{\tau^{LB}_{D(1)}(c,x)-\tau^{UB}_{D(0)}(c,x)}
    \end{equation*}

\end{theorem}

Theorem 3 provides the partial identification results for the LATE as a function of the sensitivity parameter $c$. As in Section 4, the bounds for the LATE can be used to identify the breakdown point to a particular conclusion under joint \textit{$c$-dependence}. Estimation and inference procedures for this case are similar to the ones presented in Section 6. Once again, the bounds for the unconditional LATE are identified by integrating the conditional bounds over the distribution of covariates.

When it comes to estimation, the nonparametric estimators for conditional probabilities defined in Section 6 can be used as plug-ins to build an estimator for the bounds of the LATE under joint \textit{$c$-dependence}. The proposition below shows that the estimators of the bounds converge to Gaussian elements:

\begin{proposition}
    Suppose Assumptions 1-3 and 6-9 hold. Then,

    \begin{equation*}
        \sqrt{N}\begin{pmatrix}
\widehat{\tau}^{LB}(c)-\tau^{LB}(c) \\
\widehat{\tau}^{UB}(c)-\tau^{UB}(c)
\end{pmatrix}\overset{d}{\rightarrow}\textbf{Z}_{j,\tau}(y,d,z,x,c)
    \end{equation*}
\end{proposition}

Now, turn to the breakdown point for the conclusion that the LATE is equal or greater than $\mu$, which is denoted by $c^{*}$. Let 

\begin{equation*}
    \widehat{c}^{*}=\inf\left\{ c\in\left [ 0,1 \right ]:\widehat{\tau}^{LB}(c)\leq\mu\right\}
\end{equation*}

denote the estimated breakdown point. The next result formally present a result about the asymptotic distribution of $\widehat{c^{*}}$:

\begin{theorem}
    Suppose Assumptions 1-3 and 6-9 hold. Furthermore, assume that $c\in\left [ 0,\overline{C} \right ]$. Then,

    \begin{equation*}
        \sqrt{N}\left ( \widehat{c}^{*}-c^{*} \right )\overset{d}{\rightarrow}\textbf{Z}_{j}^{BP}
    \end{equation*}

    where $\textbf{Z}^{BP}_{j}$ is a Gaussian random variable defined in Section 7 of Appendix A.
\end{theorem}

When it comes to inference, the same Bootstrap procedures introduced in Section 6.1 can be used in the case of joint \textit{$c$-dependence} to perform valid inference over the bounds for the LATE and the breakdown point.

\section{Monte Carlo Simulations}

In this section I perform Monte Carlo exercises to study the properties of the estimators for the breakdown points under the DGP described in Section 5. I consider a sample size $n$ equal to 1000 and conduct 1000 Monte Carlo simulations to study the performance of the estimators for breakdown points regarding different conclusions.

To focus the nondegenerate case, I only consider claims that yield breakdown points greater than zero.

I analyze the performance of the estimators for the breakdown points of conclusions regarding the share of compliers and the ITT under separate $c_{1}$ and \textit{$c_{2}$-dependence}, and conclusions regarding the LATE under joint \textit{$c$-dependence}. Table 1 displays the results. The estimators are analyzed in terms of their average and median bias, and the 95 \% coverage of the one-sided confidence interval from \cite{klinesantos}.

Overall, the estimators exhibit desirable finite-sample properties, expressed in terms of close to zero finite-sample bias and empirical coverage of the confidence interval being close to the target 95 \%. The performance of the estimators is stable across different values of $\mu$.

\begin{table}[]
\caption{Monte Carlo Simulations}
\centering
\begin{tabular}{cccc}
\hline
          & $c_{1}^{*}$     & $c_{2}^{*}$     & $c^{*}$      \\ \hline
$\mu=0$      & \multicolumn{3}{c}{}     \\ \hline
Av. Bias  & 0.001  & -0.003 & 0.003  \\
Med. Bias & -0.001 & -0.005 & -0.002 \\
Cover     & 0.932  & 0.927  & 0.954  \\ \hline
$\mu=0.10$   & \multicolumn{3}{c}{}     \\ \hline
Av. Bias  & 0.001  & 0.009  & 0.004  \\
Med. Bias & 0.002  & 0.008  & -0.002 \\
Cover     & 0.929  & 0.916  & 0.942  \\ \hline
$\mu=0.20$   & \multicolumn{3}{c}{}     \\ \hline
Av. Bias  & 0.011  & 0.008  & 0.008  \\
Med. Bias & 0.009  & 0.006  & 0.007  \\
Cover     & 0.931  & 0.926  & 0.933  \\ \hline
\end{tabular}
\scriptsize \noindent \\\textit{Note:} Simulations based on 1.000 Monte Carlo experiments with sample size $N=1.000$. One-sided CIs for were built using the procedure from \cite{klinesantos}.
\end{table}

\section{Empirical Application: Family Size and Employment}

In this Section, I use the estimators from Sections 6 and 7 to perform the breakdown analysis for the results regarding family size and female employment in \cite{angev}, using data from the US Census Public Use Microsamples married mothers aged 21–35 in 1980 with at least 2 children and oldest child less than 18.

In this setting, the dependent variable is and indicator for women who worked for pay in 1979. Treatment is an indicator for women having three or more children, and the instrument is an indicator for women whose first two children have the same sex. The authors control for age, age at the first birth, race and sex of the first two children as covariates.

To begin the sensitivity analysis, I use selection on observables to to calibrate the beliefs regarding the amount of selection on unobservables. I take the approach from \cite{altonji} and \cite{mastenpoirier2018}. I partition the vector of covariates $X$ as $(X_{k},X_{-k})$, where $X_{k}$ is the $k$-th component and $X_{-k}$ is a vector with remaining components. The measures used to calibrate the beliefs regarding deviations from independence are

\begin{equation*}
    \overline{c}_{k}=\sup_{x_{-k}}\sup_{x_{k}}\left | \mathbb{P}\left ( Z=1|X=(x_{k},x_{-k}) \right )-\mathbb{P}\left ( Z=1|X_{-k}=x_{-k} \right )\right |
\end{equation*}

In the data, the largest value obtained form $\overline{c}_{k}$ is associated to to the indicator for women whose second child is a man, which was estimated to be $\overline{c}_{2nd\ sex}=0.011$. 

\begin{figure}[t!]
\justify     
\caption{Family Size and Female Employment - Compliers and ITT}
\subfigure[First-stage]{\includegraphics[width=80mm]{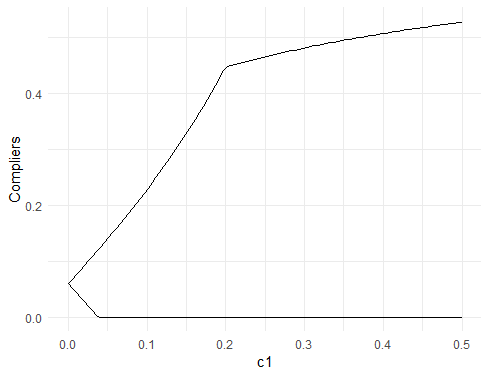}}
\subfigure[Reduced form]{\includegraphics[width=80mm]{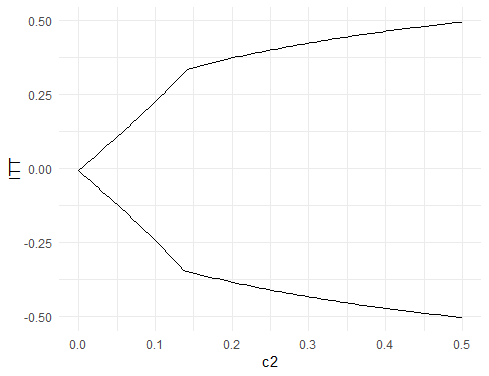}}\\

\scriptsize \noindent \textit{Note:} Figure (a) show the identified set for the share of compliers in the Angrist \& Evans (1998) setting under different values of $c_{1}$. Figure (b) show the identified set for the ITT in the Angrist \& Evans (1998) setting under different values of $c_{2}$.
\end{figure}

Picture 3 shows the identified sets for the share of compliers and the ITT in the application. The share of compliers in the case where $c_{1}=0$ is approximately 0.060 and the estimated breakdown point for the conclusion that the share of compliers is greater than zero is  $\widehat{c}_{1}^{*}=0.037$. The identified set for the ITT is displayed on the left. If point identification holds, then the ITT is equal to -0.008. However, the identified set is uninformative for most of the values of $c_{2}$. The estimated breakdown point for the conclusion that treatment effects are negative is $\widehat{c}_{2}^{*}=0.004$. However, it is not statistically different from zero.

\begin{figure}[t!]
    \justify
    \centering
    \caption{Identified-set for the effects of familiy size on unemployment}
    \includegraphics{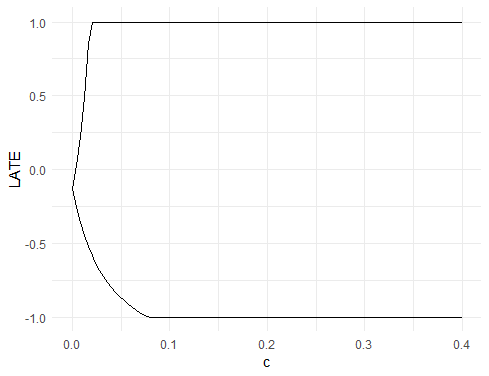}
    \label{fig:enter-label}\\
    \scriptsize \noindent \textit{Note:} Figure 4 show the identified-set for the LATE in the Angrist \& Evans (1998) setting under different values of $c$.
\end{figure}

Figure 4 shows the identified set for the LATE under joint \textit{$c$-dependence}. In the case of full independence of the instrument ($c=0)$, the LATE is equal to -0.132. The estimated breakdown point for the conclusion that the LATE is negative is $\widehat{c}^{*}=0.004$, and it is not statistically different from zero.

Overall, the results from the breakdown analysis suggest that the conclusions regarding the effects of family size on unemployment using same-sex siblings as the instrument are not robust to violations of independence of the instrument.

\section{Conclusion}

In this paper I discuss the partial identification of treatment effects in IV settings with binary outcomes under violations of independence. I derive identified sets for the first-stage, the reduced form and the LATE. Building on this result, I identify breakdown values for conclusions regarding these parameters. I derive the asymptotic properties for the estimators of the bounds and the breakdown values. I also derive sharp bounds for the LATE under a joint \textit{$c$-dependence} assumption.

Monte Carlo simulations show the desirable properties of the estimators for the breakdown points and the bootstrap procedure for constructing one-sided confidence intervals. This is still a work in progress. In the empirical application I study the effects of family size on female employment and find that weak conclusions about the share of compliers to the same-sex sibling instrument and treatment effects are highly sensitive to relaxations of random assignment of the instrument.

\bibliographystyle{apalike}
\bibliography{biblio.bib}

\section*{Appendix A}

\subsection*{Proof of Proposition 1}

 I begin with the first-stage. Consider the upper bound: $\widehat{\tau}_{D}^{UB}(c_{1},x)=\widehat{\tau}^{UB}_{D(1)}(c_{1},x)-\widehat{\tau}^{LB}_{D(0)}(c_{1},x)$. It follows directly from Lemma 4 in Appendix B that

\begin{equation*}
    \sqrt{N}\left ( \widehat{\tau}^{UB}_{D}(c_{1},x)-\tau^{UB}_{D}(c_{1},x) \right )\overset{d}{\rightarrow}\left ( \tilde{\textbf{Z}}^{(1)}_{D}(d,1,x,c_{1})-\tilde{\textbf{Z}}^{(2)}_{D}(d,0,x,c_{1}) \right )
\end{equation*}

Now we turn to the lower bound: $\widehat{\tau}^{LB}_{D}(c_{1},x)=\max\left\{ 0,\widehat{\tau}^{LB}_{D(1)}(c_{1},x)-\widehat{\tau}^{UB}_{D(0)}(c_{1},x)\right\}$. For fixed $d$ and $c_{1}$, define the mapping

\begin{align*}
    & \phi^{*}_{D}:l^{\infty}\left ( \left\{ 0,1\right\}\times\left\{ 0,1\right\}\times \mathcal{S}(X) \right )\times l^{\infty}\left ( \left\{ 0,1\right\}\times \mathcal{S}(X) \right )\times l^{\infty}\left ( \mathcal{S}(X) \right )\\&\rightarrow l^{\infty}\left ( \left\{ 0,1\right\},\mathcal{S}(X),\mathbb{R} \right )
\end{align*}

by 

\begin{equation*}
    \left [ \phi_{D}^{*}(\theta) \right ](z,x)=\max\left\{0,\theta^{(1)}(d,z,x)-\theta^{(2)}(d,z,x) \right\}
\end{equation*}

which is comprised by the functionals $\delta_{D,1}^{*}(\theta)=0$ and $\delta_{D,2}^{*}(\theta)=\theta^{(1)}(d,z,x)-\theta^{(2)}(d,z,x)$, with Hadamard derivatives respectively equal to 0 and $h^{(1)}(d,z,x)-h^{(2)}(d,z,x)$. Thus, the Hadamard derivative of $\phi_{FS}^{*}$ at $\theta_{0}$ is

\begin{equation*}
    \phi_{D,\theta_{0}}^{*'}(h)=\begin{pmatrix}
\mathbf{1}\left ( \delta_{D,1}^{*}(\theta_{0})=\delta_{D,2}^{*}(\theta_{0}) \right )\max\left\{ 0,\delta_{D,2,\theta_{0}}^{*'}(h)\right\} \\
+\mathbf{1}\left ( \delta_{D,1}^{*}(\theta_{0})>\delta_{D,2}^{*}(\theta_{0}) \right )\delta_{D,2,\theta_{0}}^{*'}(h)
\end{pmatrix}
\end{equation*}

Using the Delta Method for Hadamard differentiable functions, we obtain

\begin{equation*}
    \left [ \sqrt{N}\left ( \phi_{D}^{*}(\widehat{\theta})-\phi_{D}^{*}(\theta_{0}) \right ) \right ](z,x)\overset{d}{\rightarrow}\left [ \phi_{D,\theta_{0}}^{*'}(\tilde{\textbf{Z}}_{D}) \right ](z,x)\equiv \textbf{Z}^{*(2)}_{FS}(d,z,x,c_{1})
\end{equation*}

Combining the results yields

\begin{equation*}
    \sqrt{N}\begin{pmatrix}
\widehat{\tau}^{UB}_{D}(c_{1},x)-\tau^{UB}_{D}(c_{1},x) \\
\widehat{\tau}^{LB}_{D}(c_{1},x)-\tau^{LB}_{D}(c_{1},x)
\end{pmatrix}\overset{d}{\rightarrow}\begin{pmatrix}
\left ( \tilde{\textbf{Z}}_{D}^{(1)}(d,1,x,c_{1})-\tilde{\textbf{Z}}_{D}^{(2)}(d,0,x,c_{1}) \right ) \\
\textbf{Z}^{*(2)}_{FS}(d,z,x,c_{1})
\end{pmatrix}\equiv\textbf{Z}_{FS}^{*}(d,z,x,c_{1})
\end{equation*}

It follows that the estimators for the unconditional bounds converge weakly to Gaussian elements. We have

\begin{align*}
    &\sqrt{N}\left ( \widehat{\tau}^{LB}_{D}(c_{1})-\tau^{LB}_{D}(c_{1}) \right )\overset{d}{\rightarrow}\sum_{k=1}^{K}\left ( \tilde{\textbf{Z}}_{D}^{(1)}(d,1,x,c_{1})-\tilde{\textbf{Z}}_{D}^{(2)}(d,0,x,c_{1}) \right )+\sum_{k=1}^{K}\tau_{D}^{LB}(c_{1},x_{k})\textbf{Z}_{D}^{(3)}(0,0x_{k})\\&\equiv\textbf{Z}_{FS}^{(1)}(d,z,x,c_{1})
\end{align*}

and

\begin{equation*}
    \sqrt{N}\left ( \widehat{\tau}^{UB}_{D}(c_{1})-\tau^{UB}_{D}(c_{1}) \right )\overset{d}{\rightarrow}\sum_{k=1}^{K}\textbf{Z}^{*(2)}_{FS}(d,z,x,c_{1})+\sum_{k=1}^{K}\tau_{D}^{UB}(c_{1},x_{k})\textbf{Z}_{D}^{(3)}(0,0x_{k})\equiv\textbf{Z}_{FS}^{(2)}(d,z,x,c_{1})
\end{equation*}

Now, turn to the reduced form. I We begin with the upper bound: $\widehat{\tau}_{Y}^{UB}(c_{2},x)=\widehat{\tau}^{UB}_{Y(D(1))}(c_{2},x)-\widehat{\tau}^{LB}_{Y(D(0))}(c_{2},x)$. It follows directly from Lemma 4 that

\begin{equation*}
    \sqrt{N}\left ( \widehat{\tau}^{UB}_{Y}(c_{2},x)-\tau^{UB}_{Y}(c_{2},x) \right )\overset{d}{\rightarrow}\left ( \tilde{\textbf{Z}}^{(1)}_{Y}(y,1,x,c_{2})-\tilde{\textbf{Z}}^{(2)}_{Y}(y,0,x,c_{2}) \right )
\end{equation*}

Now we turn to the lower bound: $\widehat{\tau}^{LB}_{Y}(c_{2},x)= \widehat{\tau}^{LB}_{Y(D(1))}(c_{2},x)-\widehat{\tau}^{UB}_{Y(D(0))}(c_{2},x)$.

It also follows directly from Lemma 4 that

\begin{equation*}
    \sqrt{N}\left ( \widehat{\tau}^{LB}_{Y}(c_{2},x)-\tau^{LB}_{Y}(c_{2},x) \right )\overset{d}{\rightarrow}\left ( \tilde{\textbf{Z}}^{(2)}_{Y}(y,1,x,c_{2})-\tilde{\textbf{Z}}^{(1)}_{Y}(y,0,x,c_{2}) \right )
\end{equation*}

Hence, we have

\begin{equation*}
    \sqrt{N}\begin{pmatrix}
\widehat{\tau}^{LB}_{Y}(c_{2},x)-\tau^{LB}_{Y}(c_{2},x) \\\widehat{\tau}^{UB}_{Y}(c_{2},x)-\tau^{UB}_{Y}(c_{2},x)
\end{pmatrix}\overset{d}{\rightarrow}\begin{pmatrix}
\tilde{\textbf{Z}}_{Y}^{(1)}(y,1,x,c_{2})-\tilde{\textbf{Z}}_{Y}^{(2)}(y,0,x,c_{2}) \\\tilde{\textbf{Z}}_{Y}^{(2)}(y,1,x,c_{2})-\tilde{\textbf{Z}}_{Y}^{(1)}(y,0,x,c_{2})
\end{pmatrix}\equiv\tilde{\textbf{Z}}_{RF}^{*}(y,z,x,c_{2})
\end{equation*}

It follows that the estimators for the unconditional bounds converge weakly to Gaussian elements. We have

\begin{align*}
    &\sqrt{N}\left ( \widehat{\tau}^{LB}_{Y}(c_{2})-\tau^{LB}_{Y}(c_{2}) \right )\overset{d}{\rightarrow}\sum_{k=1}^{K}\left ( \tilde{\textbf{Z}}_{Y}^{(1)}(d,1,x,c_{2})-\tilde{\textbf{Z}}_{Y}^{(2)}(d,0,x,c_{2}) \right )+\sum_{k=1}^{K}\tau_{Y}^{LB}(c_{1},x_{k})\textbf{Z}_{Y}^{(3)}(0,0x_{k})\\&\equiv\textbf{Z}_{RF}^{(1)}(d,z,x,c_{2})
\end{align*}

and

\begin{align*}
    &\sqrt{N}\left ( \widehat{\tau}^{UB}_{Y}(c_{2})-\tau^{UB}_{Y}(c_{2}) \right )\overset{d}{\rightarrow}\sum_{k=1}^{K}\left ( \tilde{\textbf{Z}}_{Y}^{(2)}(d,1,x,c_{2})-\tilde{\textbf{Z}}_{Y}^{(1)}(d,0,x,c_{2}) \right )+\sum_{k=1}^{K}\tau_{Y}^{UB}(c_{2},x_{k})\textbf{Z}_{Y}^{(3)}(0,0x_{k})\\&\equiv\textbf{Z}_{RF}^{(2)}(d,z,x,c_{2})
\end{align*}        

which concludes the proof.

\subsection*{Proof of Theorem 1}

The breakdown point $c_{1}^{*}$ is defined implicitly by $\tau_{D}^{LB}(c_{1}^{*})=\mu$. The function $\tau^{LB}_{D}(c_{1})$ satisfies the Assumptions from van de Vaart (2000) due to Assumptions ???. Thus the mapping $\tau^{LB}_{D}(.)\mapsto c_{1}^{*}$ is Hadamard differentiable tangentially to the set of Càdlàg functions on $\left [ 0,\overline{C} \right ]$ with derivative equal to $\frac{-\partial h(c_{1}^{*})}{\frac{\partial \tau^{LB}_{D}(c_{1}^{*})}{\partial c_{1}}}$.

Recall that $\sqrt{N}\left(\widehat{\tau}_{D}^{LB}(c_{1})-\tau_{D}^{LB}(c_{1}) \right)$ converges in distribution to a random element of $l^{\infty}\left ( \left [ 0,\overline{C} \right ] \right )$ with continuous paths.

Let $\tilde{c}_{1}^{*}=\inf\left\{ c_{1}\in\left [ 0,\overline{C} \right ]: \widehat{\tau}^{LB}_{D}(c_{1})\leq\mu\right\}$. Since $c_{1}\in\left [ 0,\overline{C} \right ]$ by monotonicity of $\tau^{LB}_{D}(.)$, it follows that $\tau^{LB}_{D}(\overline{C})\leq\mu$. By $\sqrt{N}$ convergence of the First-stage bounds, we have

\begin{equation*}
    \mathbb{P}\left ( \tau_{D}^{LB}(\overline{C})\geq\mu \right )=\mathbb{P}\left ( \sqrt{N}\left ( \tau^{LB}(\overline{C})-\mu \right ) <\sqrt{N}\left ( \widehat{\tau}^{LB}_{D}(\overline{C})-\tau^{LB}_{D}(\overline{C}) \right )\right )\rightarrow 0
\end{equation*}

Therefore, the set $\left\{ c_{1}\in\left [ 0,\overline{C} \right ]:\widehat{\tau}^{LB}_{D}(c_{1})\leq\mu\right\}$ is nonempty almost surely, which implies that $\mathbb{P}\left ( \tilde{c}_{1}^{*}=\widehat{c}_{1}^{*} \right )$ approaches one as $N\rightarrow\infty$. Hence, it follows that

\begin{align*}
    &\sqrt{N}\left ( \widehat{c}_{1}^{*}-c_{1}^{*} \right )=\sqrt{N}(\widehat{c}_{1}^{*}-\tilde{c}_{1}^{*})+\sqrt{N}(\tilde{c}_{1}^{*}-c_{1}^{*})\\&o_{p}(1)+\sqrt{N}(\tilde{c}_{1}^{*}-c_{1}^{*})\overset{d}{\rightarrow}\textbf{Z}^{BP}_{FS}
\end{align*}

Applying the same logic to the reduced form yields the result for $\widehat{c}_{2}^{*}$, which concludes the proof.

\subsection*{Proof of Proposition 2}

Let $\widehat{\theta}_{\tau}=\left(\widehat{\tau}^{LB}_{Y}(c_{2}),\widehat{\tau}^{UB}_{Y}(c_{2}),\widehat{\tau}^{UB}_{D}(c_{1}),\widehat{\tau}^{LB}_{D}(c_{1})\right)$ and $\theta_{\tau}=\left( \tau^{LB}_{Y}(c_{2}),\tau^{UB}_{Y}(c_{2}),\tau^{UB}_{D}(c_{1}),\tau^{LB}_{D}(c_{1})\right)$. For fixed $y,d,c_{1},c_{2}$, define the mapping

\begin{align*}
    &\phi_{\tau}:l^{\infty}\left ( \left\{ 0,1\right\}^{3}\times\mathcal{S}(X) \right )\times l^{\infty}\left ( \left\{ 0,1\right\}^{2}\times\mathcal{S}(X) \right )\times l^{\infty}\left ( \left\{ 0,1\right\}\times\mathcal{S}(X) \right )\times l^{\infty}\left ( \mathcal{S}(X) \right )\\&\rightarrow l^{\infty}\left ( \left\{ -1,1\right\},\mathcal{S}(X),\mathbb{R}^{2} \right )
\end{align*}

by

\begin{equation*}
    \left [ \phi_{\tau}(\theta_{\tau}) \right ](z,x)=\begin{pmatrix}
\min\left\{ \frac{\theta_{\tau}^{(1)}(y,z,x,c_{2})}{\theta_{\tau}^{(3)}(d,z,x,c_{1})},1\right\} \\
\max\left\{ \frac{\theta_{\tau}^{(2)}(y,z,x,c_{2})}{\theta_{\tau}^{(4)}(d,z,x,c_{1})},-1\right\}
\end{pmatrix}
\end{equation*}

The Hadamard derivative for $\left [ \delta_{1}(\theta) \right ](z,x)=\frac{\theta^{(1)}(y,z,x,c_{2})}{\theta^{(3)}(d,z,x,c_{1})}$ is equal to

\begin{equation*}
    \left [ \delta^{'}_{1,\theta}(h) \right ](z,x)=\frac{h^{(1)}(y,z,x,c_{2})\theta^{(3)}(d,z,x,c_{1})-\theta^{(1)}(y,z,x,c_{2})h^{(3)}(d,z,x,c_{1})}{\left ( \theta^{(3)}(d,z,x,c_{1}) \right )^{2}}
\end{equation*}

The Hadamard derivative for $\left [ \delta_{2}(\theta) \right ](z,x)=1$ is equal to $0$. The hadamard derivative for $\left [ \delta_{3}(\theta) \right ](z,x)=\frac{\theta^{(2)}(y,z,x,c_{2})}{\theta^{(4)}(d,z,x,c_{1})}$ is equal to

\begin{equation*}
    \left [ \delta^{'}_{3,\theta}(h) \right ](z,x)=\frac{h^{(2)}(y,z,x,c_{2})\theta^{(4)}(d,z,x,c_{1})-\theta^{(2)}(y,z,x,c_{2})h^{(4)}(d,z,x,c_{1})}{\left ( \theta^{(4)}(d,z,x,c_{1}) \right )^{2}}
\end{equation*}

And the Hadamard derivative for $\left [ \delta_{4}(\theta) \right ](z,x)=-1$ is equal to $0$.

Hence, the Hadamard directional derivative of $\phi_{\tau}$ evaluated at $\theta_{\tau,0}$ is 

\begin{equation*}
    \phi_{\tau,\theta_{0}}^{'}(h)=\begin{pmatrix}
\mathbf{1}\left ( \delta_{\tau,1}(\theta_{0})=1 \right )\max\left\{ \delta^{'}_{\tau1,\theta_{0}}(h),0\right\} \\
+\mathbf{1}\left ( \delta_{\tau,1}(\theta_{0})<1 \right )\delta^{'}_{\tau1,\theta_{0}}(h) \\ 
 \\\mathbf{1}\left ( \delta_{\tau,3}(\theta_{0})=-1 \right )\min\left\{ \delta^{'}_{\tau3,\theta_{0}}(h),0\right\}
 \\+\mathbf{1}\left ( \delta_{\tau,3}(\theta_{0})>-1 \right )\delta^{'}_{\tau3,\theta_{0}}(h) 
\end{pmatrix}
\end{equation*}

By Proposition 1 and the Delta Method for Hadamard directionally differentiable functions,

\begin{equation*}
    \sqrt{N}\left ( \phi_{\tau}(\widehat{\theta})-\phi_{\tau}(\theta_{0}) \right )\overset{d}{\rightarrow}\left [ \theta_{\tau,\theta_{0}}^{'}(\tilde{\textbf{Z}}) \right ](z,x)\equiv\tilde{\textbf{Z}}_{\tau}(z,x)
\end{equation*}

which yields the process $\textbf{Z}_{\tau}(y,d,z,x,c_{1},c_{2})$.

\subsection*{Proof of Theorem 2}

By Lemmas 3 and 4 and Proposition 1, the numerator converges uniformly over $\mathcal{C}\times\mathcal{M}$. By Proposition 1, the denominator also converges uniformly. Hence, applying the Delta Method we obtain

\begin{equation*}
    \sqrt{N}\left ( \widehat{BF}(c_{1},\mu) -BF(c_{1},\mu)\right )\overset{d}{\rightarrow}\textbf{Z}_{BF}(c_{1},\mu)
\end{equation*}

which concludes the proof.

\subsection*{Proof of Theorem 3}

I begin by deriving the identified set for the joint probabilities $ \mathbb{P}\left(Y(D(z))=1,D(z)=1|X=x\right)$ and $\mathbb{P}\left(Y(D(z))=1,D(z)=0|X=x\right)$ under joint \textit{$c$-dependence}. The bounds on this joint probabilities are subsequently combined to obtain the bounds for potential quantities. I begin with $\mathbb{P}\left(Y(D(z))=1,D(z)=1|X=x\right)$.

First, consider the upper bound. We have

\begin{align*}
    &\mathbb{P}\left(Y(D(z))=1,D(z)=1|X=x\right)\\&=p_{Y,D|z,x}p_{z|x}+\mathbb{P}\left(Y(D(z))=1,D(z)=1|Z=1-z,X=x\right)p_{1-z|x}\\&\leq p_{Y,D|z,x}p_{z|x}+(1-p_{z|x})
\end{align*}

Moreover, we have

\begin{align*}
    &\mathbb{P}\left ( Y(D(z))=1,D(z)=1|X=x \right )=\frac{\mathbb{P}\left ( Y(D(z))=1,D(z)=1,X=x \right )}{\mathbb{P}\left ( X=x \right )}\\&\cdot\frac{\mathbb{P}\left ( Y(D(z))=1,D(z)=1,Z=z,X=x \right )\mathbb{P}\left ( Z=z,X=x \right )}{\mathbb{P}\left ( Y(D(z))=1,D(z)=1,Z=z,X=x \right )\mathbb{P}\left ( Z=z,X=x \right )}\\&=\frac{p_{Y,D|z,x}p_{z|x}}{\mathbb{P}\left ( Z=z|Y(D(z))=1,D(z)=1,X=x\right )}\\&\leq \frac{p_{Y,D|z,x}p_{z|x}}{p_{z|x}-c}
\end{align*}

where the last line follows from joint \textit{$c$-dependence}.

For the lower bound, we have

\begin{align*}
    &\mathbb{P}\left(Y(D(z))=1,D(z)=1|X=x\right)\\&=p_{Y,D|z,x}p_{z|x}+\mathbb{P}\left(Y(D(z))=1,D(z)=1|Z=1-z,X=x\right)p_{1-z|x}\\&\geq p_{Y,D|z,x}p_{z|x}
\end{align*}

and

\begin{align*}
    &\mathbb{P}\left ( Y(D(z))=1,D(z)=1|X=x \right )=\frac{\mathbb{P}\left ( Y(D(z))=1,D(z)=1,X=x \right )}{\mathbb{P}\left ( X=x \right )}\\&\cdot\frac{\mathbb{P}\left ( Y(D(z))=1,D(z)=1,Z=z,X=x \right )\mathbb{P}\left ( Z=z,X=x \right )}{\mathbb{P}\left ( Y(D(z))=1,D(z)=1,Z=z,X=x \right )\mathbb{P}\left ( Z=z,X=x \right )}\\&=\frac{p_{Y,D|z,x}p_{z|x}}{\mathbb{P}\left ( Z=z|Y(D(z))=1,D(z)=1,X=x\right )}\\&\geq \frac{p_{Y,D|z,x}p_{z|x}}{p_{z|x}+c}
\end{align*}

From the lower bound presented above, it follows that

\begin{align*}
    &\mathbb{P}\left ( Y(D(z))=1,D(z)=1|X=x \right )=1-\mathbb{P}\left ( Y(D(z))=1,D(z)=0|X=x \right )\\&-\mathbb{P}\left ( Y(D(z))=0,D(z)=1|X=x \right )-\mathbb{P}\left ( Y(D(z))=0,D(z)=0|X=x \right )\\&=1-\frac{p_{Y,1-D|z,x}p_{z|x}}{\mathbb{P}\left ( Z=z|Y(D(z))=1,D(z)=0,X=x \right )}\\&-\frac{p_{1-Y,D|z,x}p_{z|x}}{\mathbb{P}\left ( Z=z|Y(D(z))=0,D(z)=1,X=x \right )}-\frac{p_{1-Y,1-D|z,x}p_{z|x}}{\mathbb{P}\left ( Z=z|Y(D(z))=0,D(z)=0,X=x \right )}\\&\leq 1-\frac{p_{Y,1-D|z,x}p_{z|x}}{p_{z|x}+c}-\frac{p_{1-Y,D|z,x}p_{z|x}}{p_{z|x}+c}-\frac{p_{1-Y,1-D|z,x}p_{z|x}}{p_{z|x}+c}\\&=\frac{p_{Y,D|z,x}p_{z|x}+c}{p_{z|x}+c}
\end{align*}

Also,

\begin{align*}
    &\mathbb{P}\left ( Y(D(z))=1,D(z)=1|X=x \right )=1-\mathbb{P}\left ( Y(D(z))=1,D(z)=0|X=x \right )\\&-\mathbb{P}\left ( Y(D(z))=0,D(z)=1|X=x \right )-\mathbb{P}\left ( Y(D(z))=0,D(z)=0|X=x \right )\\&=1-\frac{p_{Y,1-D|z,x}p_{z|x}}{\mathbb{P}\left ( Z=z|Y(D(z))=1,D(z)=0,X=x \right )}\\&-\frac{p_{1-Y,D|z,x}p_{z|x}}{\mathbb{P}\left ( Z=z|Y(D(z))=0,D(z)=1,X=x \right )}-\frac{p_{1-Y,1-D|z,x}p_{z|x}}{\mathbb{P}\left ( Z=z|Y(D(z))=0,D(z)=0,X=x \right )}\\&\geq 1-\frac{p_{Y,1-D|z,x}p_{z|x}}{p_{z|x}-c}-\frac{p_{1-Y,D|z,x}p_{z|x}}{p_{z|x}-c}-\frac{p_{1-Y,1-D|z,x}p_{z|x}}{p_{z|x}-c}\\&=\frac{p_{Y,D|z,x}p_{z|x}-c}{p_{z|x}-c}
\end{align*}

Therefore, it follows that 

\begin{equation*}
    \mathbb{P}\left ( Y(D(z))=1,D(z)=1|X=x \right )\geq\max\left\{ \frac{p_{Y,D|z,x}p_{z|x}}{p_{z|x}+c}, \frac{p_{Y,D|z,x}p_{z|x}-c}{p_{z|x}-c},p_{Y,D|z,x}p_{z|x}\right\}
\end{equation*}

and

\begin{equation*}
    \mathbb{P}\left ( Y(D(z))=1,D(z)=1|X=x \right )\leq\min\left\{ \frac{p_{Y,D|z,x}p_{z|x}}{p_{z|x}-c}, \frac{p_{Y,D|z,x}p_{z|x}+c}{p_{z|x}+c},p_{Y,D|z,x}p_{z|x}+(1-p_{z|x})\right\}
\end{equation*}

Similarly, one can show that

\begin{equation*}
    \mathbb{P}\left ( Y(D(z))=1,D(z)=0|X=x \right )\geq\max\left\{ \frac{p_{Y,1-D|z,x}p_{z|x}}{p_{z|x}+c}, \frac{p_{Y,1-D|z,x}p_{z|x}-c}{p_{z|x}-c},p_{Y,1-D|z,x}p_{z|x}\right\}
\end{equation*}

and

\begin{equation*}
    \mathbb{P}\left ( Y(D(z))=1,D(z)=0|X=x \right )\leq\min\left\{ \frac{p_{Y,1-D|z,x}p_{z|x}}{p_{z|x}-c}, \frac{p_{Y,1-D|z,x}p_{z|x}+c}{p_{z|x}+c},p_{Y,1-D|z,x}p_{z|x}+(1-p_{z|x})\right\}
\end{equation*}

Now, I build on these bounds to identify the bounds on potential quantities. I derive the results for $\tau_{Y(D(z))}(x)$, the results for $\tau_{D(z)}(x)$ are analogous.

I begin with the upper bound. First, note that 

\begin{align*}
    &\mathbb{P}\left(Y(D(z))=1|X=x\right)=p_{Y|z,x}p_{z|x}+\mathbb{P}\left(Y(D(z))=1|Z=0,X=x\right)(1-p_{z|x})\\&\leq p_{Y|z,x}p_{z|x}+(1-p_{z|x})
\end{align*}

Furthermore, note that by the Law of Total Probability, 

\begin{align*}
    &\mathbb{P}\left ( Y(D(z))=1|X=x \right )=\mathbb{P}\left ( Y(D(z))=1,D(z)=1|X=x \right )+\mathbb{P}\left ( Y(D(z))=1,D(z)=0|X=x \right )\\&\leq\frac{p_{Y,D|z,x}p_{z|x}}{p_{z|x}-c}+\frac{p_{Y,1-D|z,x}p_{z|x}}{p_{z|x}-c}=\frac{p_{Y|z,x}p_{z|x}}{p_{z|x}-c}
\end{align*}

and

\begin{align*}
    &\mathbb{P}\left ( Y(D(z))=1|X=x \right )=\mathbb{P}\left ( Y(D(z))=1,D(z)=1|X=x \right )+\mathbb{P}\left ( Y(D(z))=1,D(z)=0|X=x \right )\\&\leq\frac{p_{Y,D|z,x}p_{z|x}+c}{p_{z|x}+c}+\frac{p_{Y,1-D|z,x}p_{z|x}+c}{p_{z|x}+c}=\frac{p_{Y|z,x}p_{z|x}+2c}{p_{z|x}+c}
\end{align*}

Combining the inequalities yields the upper bound. Now consider the lower bound. First, we have

\begin{align*}
    &\mathbb{P}\left(Y(D(z))=1|X=x\right)=p_{Y|z,x}p_{z|x}+\mathbb{P}\left(Y(D(z))=1|Z=0,X=x\right)(1-p_{z|x})\\&\geq p_{Y|z,x}p_{z|x}
\end{align*}

Furthermore, note that by the Law of Total Probability, 

\begin{align*}
    &\mathbb{P}\left ( Y(D(z))=1|X=x \right )=\mathbb{P}\left ( Y(D(z))=1,D(z)=1|X=x \right )+\mathbb{P}\left ( Y(D(z))=1,D(z)=0|X=x \right )\\&\geq\frac{p_{Y,D|z,x}p_{z|x}}{p_{z|x}+c}+\frac{p_{Y,1-D|z,x}p_{z|x}}{p_{z|x}+c}=\frac{p_{Y|z,x}p_{z|x}}{p_{z|x}+c}
\end{align*}

and

\begin{align*}
    &\mathbb{P}\left ( Y(D(z))=1|X=x \right )=\mathbb{P}\left ( Y(D(z))=1,D(z)=1|X=x \right )+\mathbb{P}\left ( Y(D(z))=1,D(z)=0|X=x \right )\\&\geq\frac{p_{Y,D|z,x}p_{z|x}-c}{p_{z|x}-c}+\frac{p_{Y,1-D|z,x}p_{z|x}-c}{p_{z|x}-c}=\frac{p_{Y|z,x}p_{z|x}-2c}{p_{z|x}-c}
\end{align*}

Combining the inequalities yields the lower bound, which concludes the proof for the bounds.

For sharpness, I show the proof for $\tau_{Y(D(1))}(x)$. The proof for other potential quantities is analogous. Fix $\tau^{*}\in\left[\tau_{Y(D(1))}^{LB}(c,x),\tau_{Y(D(1))}^{UB}(c,x)\right]$. Note that $\mathbb{P}\left(Y(D(1))=1|Z=1,X=x\right)$ is readily identified by the data. Hence, to prove sharpness we must chose a value for $\mathbb{P}\left(Y(D(1))=1|Z=0,X=x\right)$ such that 1) $\mathbb{P}\left(Y(D(1))=1|X=x\right)=\tau^{*}$, 2) the probability $\mathbb{P}\left(Y(D(1))=1|Z=0,X=x\right)$ is well defined and 3) joint \textit{$c$-dependence} holds.

First, define $\mathbb{P}\left ( Y(D(1))=1,D(1)=1|Z=0,X=x \right )=\alpha\frac{\tau^{*}-p_{Y|1,x}p_{1|x}}{p_{0|x}}$ for some $\alpha\in\left[0,1\right]$ and define $\mathbb{P}\left ( Y(D(1))=1,D(1)=0|Z=0,X=x \right )=(1-\alpha)\frac{\tau^{*}-p_{Y|1,x}p_{1|x}}{p_{0|x}}$ such that we obtain $\mathbb{P}\left ( Y(D(1))=1|Z=0,X=x \right )=\frac{\tau^{*}-p_{Y|1,x}p_{1|x}}{p_{0|x}}$ by the Law of Total Probabilities. It follows from the second part of Theorem 5 from \cite{mastenpoirier2018} that 1) and 2) hold.

To prove that joint \textit{$c$-dependence} holds, note that $\mathbb{P}\left(Z=1|Y(D(1))=1,D(1)=1,X=x\right)$ can be written as

{\small \begin{align*}
    &\mathbb{P}\left(Z=1|Y(D(1))=1,D(1)=1,X=x\right)=\frac{\mathbb{P}\left ( Y(D(1))=1|X=x \right )}{\mathbb{P}\left ( Y(D(1))=1,D(1)=1|X=x \right )}\\&\times\left [ \mathbb{P}\left ( Z=1|Y(D(1))=1,X=x \right )-\frac{\mathbb{P}\left ( Y(D(1))=1,D(1)=0|X=x \right )}{\mathbb{P}\left ( D(1)=1|X=x \right )}\mathbb{P}\left ( Z=1|Y(D(1))=1,D(1)=0,X=X \right ) \right ]
\end{align*}}

Using Bayes rule, the expression can be written as

\begin{align*}
    &\mathbb{P}\left ( Z=1|Y(D(1))=1,D(1)=1,X=x \right )=\frac{(p_{Y|1,x}-p_{Y,(1-D)|1,x})p_{1|x}}{\mathbb{P}\left ( Y(D(1))=1,D(1)=1|X=x \right )}\\&=\frac{p_{Y,D|1,x}p_{1|x}}{\mathbb{P}\left ( Y(D(1))=1,D(1)=1|X=x \right )}
\end{align*}

Note that

\begin{align*}
    &\frac{p_{Y,D|1,x}p_{1|x}}{\mathbb{P}\left ( Y(D(1))=1,D(1)=1|X=x \right )}\\&\geq\frac{p_{Y,D|1,x}p_{1|x}}{\min\left\{ \frac{p_{Y,D|1,x}p_{1|x}}{p_{1|x}-c},\frac{p_{Y,D|1,x}p_{1|x}+c}{p_{1|x}+c},p_{Y,D|1,x}p_{1|x}+p_{0|x}\right\}}\geq p_{1|x}-c
\end{align*}

and

\begin{align*}
    &\frac{p_{Y,D|1,x}p_{1|x}}{\mathbb{P}\left ( Y(D(1))=1,D(1)=1|X=x \right )}\\&\leq\frac{p_{Y,D|1,x}p_{1|x}}{\max\left\{ \frac{p_{Y,D|1,x}p_{1|x}}{p_{1|x}+c},\frac{p_{Y,D|1,x}p_{1|x}-c}{p_{1|x}-c},p_{Y,D|1,x}p_{1|x}\right\}}\leq \frac{p_{Y,D|1,x}p_{1|x}(p_{1|x}+c)}{p_{Y,D|1,x}p_{1|x}}=p_{1|x}+c
\end{align*}

which concludes the proof of sharpness.

\subsection*{Proof of Proposition 3}

Define $\tau^{UB}_{D}(c,x)=\tau^{UB}_{D(1)}(c,x)-\tau^{LB}_{D(0)}(c,x)$, $\tau^{LB}_{D}(c,x)=\tau^{LB}_{D(1)}(c,x)-\tau^{UB}_{D(0)}(c,x)$, $\tau^{UB}_{Y}(c,x)=\tau^{UB}_{Y(D(1))}(c,x)-\tau^{LB}_{Y(D(0))}(c,x)$ and $\tau^{LB}_{Y}(c,x)=\tau^{LB}_{Y(D(1))}(c,x)-\tau^{UB}_{Y(D(0))}(c,x)$. Furthemore, define $\widehat{\tau}^{UB}_{D}(c,x),\widehat{\tau}^{LB}_{D}(c,x),\widehat{\tau}^{UB}_{Y}(c,x),\widehat{\tau}^{LB}_{Y}(c,x)$ as the estimators for these quantities, respectively. It follows from Lemma 6 that

\begin{equation*}
    \sqrt{N}\begin{pmatrix}
\widehat{\tau}^{UB}_{D}(c,x)-\tau^{UB}_{D}(c,x) \\\widehat{\tau}^{LB}_{D}(c,x)-\tau^{LB}_{D}(c,x)
 \\\widehat{\tau}^{UB}_{Y}(c,x)-\tau^{UB}_{Y}(c,x)
 \\\widehat{\tau}^{LB}_{Y}(c,x)-\tau^{LB}_{Y}(c,x)
\end{pmatrix}\overset{d}{\rightarrow}\begin{pmatrix}
\textbf{Z}_{j,FS}^{(1)}(d,1,x,c)-\textbf{Z}_{j,FS}^{(2)}(d,0,x,c) \\ \textbf{Z}_{j,FS}^{(2)}(d,1,x,c)-\textbf{Z}_{j,FS}^{(1)}(d,0,x,c)
 \\ \textbf{Z}_{j,RF}^{(1)}(d,1,x,c)-\textbf{Z}_{j,RF}^{(2)}(d,0,x,c)
 \\\textbf{Z}_{j,RF}^{(2)}(d,1,x,c)-\textbf{Z}_{j,RF}^{(1)}(d,0,x,c)
\end{pmatrix}
\end{equation*}

And thus,

\begin{align*}
    &\sqrt{N}\begin{pmatrix}
\widehat{\tau}^{UB}_{D}(c)-\tau^{UB}_{D}(c) \\\widehat{\tau}^{LB}_{D}(c)-\tau^{LB}_{D}(c)
 \\\widehat{\tau}^{UB}_{Y}(c)-\tau^{UB}_{Y}(c)
 \\\widehat{\tau}^{LB}_{Y}(c)-\tau^{LB}_{Y}(c)
\end{pmatrix}\\&\overset{d}{\rightarrow}\begin{pmatrix}
\sum_{k=1}^{K}q_{x_{k}}\left ( \textbf{Z}_{j,FS}^{(1)}(d,1,x_{k},c)-\textbf{Z}_{j,FS}^{(2)}(d,0,x_{k},c) \right )+\sum_{k=1}^{K}\tau^{UB}_{D}(c,x_{k})\textbf{Z}_{j}^{(3)}(0,0,0,x_{k})  \\ \sum_{k=1}^{K}q_{x_{k}}\left ( \textbf{Z}_{j,FS}^{(2)}(d,1,x_{k},c)-\textbf{Z}_{j,FS}^{(1)}(d,0,x_{k},c) \right )+\sum_{k=1}^{K}\tau^{LB}_{D}(c,x_{k})\textbf{Z}_{j}^{(3)}(0,0,0,x_{k})
 \\ \sum_{k=1}^{K}q_{x_{k}}\left ( \textbf{Z}_{j,RF}^{(1)}(y,1,x_{k},c)-\textbf{Z}_{j,RF}^{(2)}(y,0,x_{k},c) \right )+\sum_{k=1}^{K}\tau^{UB}_{Y}(c,x_{k})\textbf{Z}_{j}^{(3)}(0,0,0,x_{k})
 \\ \sum_{k=1}^{K}q_{x_{k}}\left ( \textbf{Z}_{j,RF}^{(1)}(y,1,x_{k},c)-\textbf{Z}_{j,RF}^{(2)}(y,0,x_{k},c) \right )+\sum_{k=1}^{K}\tau^{LB}_{Y}(c,x_{k})\textbf{Z}_{j}^{(3)}(0,0,0,x_{k})
\end{pmatrix}\\&\equiv\textbf{Z}_{j}^{*}(y,d,z,x,c)
\end{align*}

Let $\widehat{\Omega}_{\tau}=\left ( \widehat{\tau}^{UB}_{Y}(c),\widehat{\tau}_{Y}^{LB}(c), \widehat{\tau}^{UB}_{D}(c),\widehat{\tau}_{D}^{LB}(c) \right )$ and $\Omega_{\tau}=\left ( \tau^{UB}_{Y}(c),\tau_{Y}^{LB}(c), \tau^{UB}_{D}(c),\tau_{D}^{LB}(c) \right )$. For fixed $y$, $d$ and $c$, define the mapping

\begin{align*}
    &\psi_{\tau}:l^{\infty}\left ( \left\{ 0,1\right\}^{3}\times\mathcal{S}(X) \right )\times l^{\infty}\left ( \left\{ 0,1\right\}^{2}\times\mathcal{S}(X) \right )\times l^{\infty}\left ( \left\{ 0,1\right\}\times\mathcal{S}(X) \right )\times l^{\infty}\left ( \mathcal{S}(X) \right )\\&\rightarrow l^{\infty}\left ( \left\{ -1,1\right\},\mathcal{S}(X),\mathbb{R}^{2} \right )
\end{align*}

by

\begin{equation*}
    \left [ \psi_{\tau}(\Omega_{\tau}) \right ](z,x)=\begin{pmatrix}
\min\left\{ \frac{\Omega^{(1)}(y,z,x)}{\Omega^{(3)}(d,z,x)},1\right\} \\\max\left\{ \frac{\Omega^{(2)}(y,z,x)}{\Omega^{(4)}(d,z,x)},-1\right\}
\end{pmatrix}
\end{equation*}

The mapping is comprise with four elements. We have

\begin{equation*}
    \left [ \delta_{\tau,1}(\Omega) \right ]=\frac{\Omega^{(2)}(y,z,x)}{\Omega^{(4)}(d,z,x)}
\end{equation*}

with Hadamard derivative equal to

\begin{equation*}
    \left [ \delta^{'}_{\tau,1,\Omega}(h) \right ]=\frac{h^{(2)}(y,z,x)}{\Omega^{(4)}(d,z,x)}-\frac{h^{(4)}(d,z,x)\Omega^{(2)}(y,z,x)}{\left ( \Omega^{(4)}(d,z,x) \right )^{2}},
\end{equation*}

\begin{equation*}
    \left [ \delta_{\tau,2}(\Omega) \right ]=-1
\end{equation*}

with Hadamard derivative equal to

\begin{equation*}
    \left [ \delta^{'}_{\tau,2,\Omega}(h) \right ]=0,
\end{equation*}

\begin{equation*}
    \left [ \delta_{\tau,3}(\Omega) \right ]=\frac{\Omega^{(1)}(y,z,x)}{\Omega^{(3)}(d,z,x)}
\end{equation*}

with Hadamard derivative equal to

\begin{equation*}
    \left [ \delta^{'}_{\tau,3,\Omega}(h) \right ]=\frac{h^{(1)}(y,z,x)}{\Omega^{(3)}(d,z,x)}-\frac{h^{(3)}(d,z,x)\Omega^{(1)}(y,z,x)}{\left ( \Omega^{(3)}(d,z,x) \right )^{2}},
\end{equation*}

and

\begin{equation*}
    \left [ \delta_{\tau,4}(\Omega) \right ]=1
\end{equation*}

with Hadamard derivative equal to

\begin{equation*}
    \left [ \delta^{'}_{\tau,4,\Omega}(h) \right ]=0,
\end{equation*}

Therefore, the Hadamard directional derivate of $\psi_{\tau}$ evaluated at $\Omega_{\tau}$ is

\begin{equation*}
    \psi_{\tau,\Omega_{\tau}}^{'}(h)=\begin{pmatrix}
\mathbf{1}\left ( \delta_{3}(\Omega_{\tau})=1) \right )\min\left\{ \delta_{3,\tau,\Omega_\tau}^{'}(h),0\right\} \\+\mathbf{1}\left ( \delta_{3}(\Omega_{\tau}>1  \right )\delta_{3,\tau,\Omega_{\tau}}^{'}(h) 
 \\
 \\\mathbf{1}\left ( \delta_{1}(\Omega_{\tau})=-1) \right )\min\left\{ \delta_{1,\tau,\Omega_\tau}^{'}(h),0\right\} 
 \\+\mathbf{1}\left ( \delta_{1}(\Omega_{\tau}>-1  \right )\delta_{1,\tau,\Omega_{\tau}}^{'}(h) 
\end{pmatrix}
\end{equation*}

It follows from the Delta Method that

\begin{equation*}
    \left [ \sqrt{N}\left ( \psi_{\tau}(\widehat{\Omega})-\psi_{\tau}(\Omega_{\tau}) \right ) \right ](z,x)\overset{d}{\rightarrow}\left [ \psi_{\tau,\Omega_{\tau}}^{'}(\textbf{Z}_{j}^{*}) \right ](z,x)\equiv\textbf{Z}_{j,\tau}(z,x)
\end{equation*}

which yields the process $\textbf{Z}_{j,\tau}(y,d,z,x,c)$

\subsection*{Proof of Theorem 4}

The steps of the proof are the same as in the proof of Theorem 1.

\section*{Appendix B}

In this section I derive the asymptotic properties of the estimators used as plug-ins for the estimations of the bounds. I show that these estimators converge uniformly to mean-zero Gaussian processes.

\begin{lemma}
    Suppose Assumptions  6 and 7 hold. Then, 

    \begin{equation*}
        \sqrt{N}\begin{pmatrix}
\widehat{p}_{D|z,x}-p_{D|z,x} \\
\widehat{p}_{z|x}-p_{z|x} \\\widehat{q}_{x}-q_{x}
\end{pmatrix}\overset{d}{\rightarrow}\textbf{Z}_{D}(d,z,x)
    \end{equation*}

    and

\begin{equation*}
    \sqrt{N}\begin{pmatrix}
\widehat{p}_{Y|z,x}-p_{Y|z,x} \\
\widehat{p}_{z|x}-p_{z|x} \\\widehat{q}_{x}-q_{x}
\end{pmatrix}\overset{d}{\rightarrow}\textbf{Z}_{Y}(y,z,x)
\end{equation*}

which are mean-zero Gaussian Processes in $l^{\infty}\left ( \left\{ 0,1\right\}\times\left\{ 0,1\right\}\times\mathcal{S}(X),\mathbb{R}^{3} \right )$ with covariance kernels respectively equal to $\Sigma_{FS}=\mathbb{E}\left [ \textbf{Z}_{D}(d,z,x)\textbf{Z}_{D}(d,\tilde{z},\tilde{x})^{'}  \right ]$ and $\Sigma_{RF}=\mathbb{E}\left [ \textbf{Z}_{Y}(y,z,x)\textbf{Z}_{Y}(y,\tilde{z},\tilde{x})^{'}  \right ]$
    
\end{lemma}

\begin{proof}
    By a second-order Taylor Expansion,

\begin{align*}
    \widehat{p}_{D|z,x}-p_{D|z,x}=\frac{\frac{1}{N}\sum_{i=1}^{N}\mathbf{1}\left ( D_{i}=1 \right )\mathbf{1}\left ( Z_{i}=z,X_{i}=x \right )}{\frac{1}{N}\sum_{i=1}^{N}\mathbf{1}\left ( Z_{i}=z,X_{i}=x \right )}-\frac{\mathbb{P}\left ( D_{i}=1,Z_{i}=z,X_{i}=x \right )}{\mathbb{P}\left ( Z_{i}=z,X_{i}=x \right )}\\=\frac{\frac{1}{N}\sum_{i=1}^{N}\mathbf{1}\left ( D_{i}=1 \right )\mathbf{1}\left ( Z_{i}=z,X_{i}=x \right )-\mathbb{P}\left ( D_{i}=1,Z_{i}=z,X_{i}=x \right )}{\mathbb{P}\left ( Z_{i}=z,X_{i}=x \right )}\\-\frac{p_{D|z,x}}{\mathbb{P}\left ( Z_{i}=z,X_{i}=x \right )}\left ( \frac{1}{N}\sum_{i=1}^{N}\mathbf{1}\left ( Z_{i}=z,X_{i}=x \right ) -\mathbb{P}\left ( Z_{i}=z,X_{i}=x \right )\right )\\+O_{p} [ \left ( \frac{1}{N}\sum_{i=1}^{N}\mathbf{1}\left ( D_{i}=1 \right )\mathbf{1}\left ( Z_{i}=z,X_{i}=x \right )-p_{D|z,x}\mathbb{P}\left ( Z_{i}=z,X_{i}=x \right ) \right ) \\\cdot \left ( \frac{1}{N}\sum_{i=1}^{N}\mathbf{1}\left ( Z_{i}=z,X_{i}=x \right )-\mathbb{P}\left ( Z_{i}=z,X_{i}=x \right ) \right ) ]\\+O_{p}\left [ \left ( \frac{1}{N}\sum_{i=1}^{N} \mathbf{1}\left ( Z_{i}=z,X_{i}=x \right )-\mathbb{P}\left ( Z_{i}=z,X_{i}=x \right )\right )^{2} \right ]
\end{align*}

By standard bracket entropy results (van der Vaart, 2000), the functions classes given by $\left\{ \mathbf{1}\left ( D_{i}=1 \right )\mathbf{1}\left ( Z_{i}=z,X_{i}=x \right ):z\in\left\{ 0,1\right\},x\in\mathcal{S}(X)\right\}$ and $\left\{ \mathbf{1}\left ( Z_{i}=z,X_{i}=x \right ):z\in\left\{ 0,1\right\},x\in\mathcal{S}(X)\right\}$ are both P-Donsker. Hence, the residuals of order $O_{p}(N^{-1})$ uniformly over $\left ( d,z,x \right )\in 1\times\left\{ 0,1\right\}\times\mathcal{S}(X)$. Using Slutkys's Theorem, we obtain the following asymptotically linear representation:

\begin{equation*}
    \widehat{p}_{D|z,x}-p_{D|z,x}=\frac{1}{N}\frac{\sum_{i=1}^{N}\left ( \mathbf{1}\left ( Z_{i}=z,X_{i}=x \right )\left ( \mathbf{1}\left ( D_{i}=1 \right )-p_{D|Z,X} \right ) \right )}{\mathbb{P}\left ( Z_{i}=z,X_{i}=x \right )}+o_{p}(N^{-1/2})
\end{equation*}

Using the same bracket entropy arguments, it follows that the linear representation is also P-Donsker. Hence, $\sqrt{N}\left(\widehat{p}_{D|z,x}-p_{D|z,x} \right)$ converges to a mean-zero gaussian process with continuous paths. Using similar arguments, we obtain

\begin{equation*}
    \widehat{p}_{z|x}-p_{z|x}=\frac{1}{N}\frac{\sum_{i=1}^{N}\mathbf{1}\left ( X_{i}=x \right )\left ( \mathbf{1}\left ( Z_{i}=z \right )-p_{z|x} \right )}{q_{x}}+o_{p}(N^{-1/2})
\end{equation*}

and

\begin{equation*}
    \widehat{q}_{x}-q_{x}=\frac{1}{N}\sum_{i=1}^{N}\mathbf{1}\left ( X_{i}=x \right )-q_{x}+o_{p}(N^{-1/2})
\end{equation*}

The covariance kernel $\Sigma_{FS}$ is calculated as follows

\begin{align*}
    &\left [ \Sigma_{FS}(z,x,\tilde{z},\tilde{x}) \right ]_{1,1}\\&=\mathbb{E}\left [ \frac{\mathbf{1}\left ( Z_{i}=z,X_{i}=x \right )\mathbf{1}\left ( Z_{i}=\tilde{z},X_{i}=\tilde{x} \right )\left ( \mathbf{1}\left ( D_{i}=1 \right )-p_{D|z,x} \right )\left ( \mathbf{1}\left ( D_{i}=1 \right )-p_{D|\tilde{z},\tilde{x}} \right )}{\mathbb{P}\left ( Z_{i}=z,X_{i}=x \right )\mathbb{P}\left ( Z_{i}=\tilde{z},X_{i}=\tilde{x} \right )} \right ],
\end{align*}

\begin{equation*}
    \left [ \Sigma_{FS}(z,x,\tilde{z},\tilde{x}) \right ]_{1,2}=\mathbb{E}\left [ \frac{\mathbf{1}\left ( Z_{i}=z,X_{i}=\tilde{x} \right )\left ( \mathbf{1}\left ( Z_{i}=\tilde{z} \right )-p_{\tilde{z}|\tilde{x}} \right )\left ( \mathbf{1}\left ( D_{i}=1 \right )-p_{D|z,x} \right )}{p_{z|X}q_{x}q_{\tilde{x}}} \right ]=0,
\end{equation*}

\begin{equation*}
    \left [ \Sigma_{FS}(z,x,\tilde{z},\tilde{x}) \right ]_{1,3}=\mathbb{E}\left [ \frac{\left ( \mathbf{1}\left ( X_{i}=\tilde{x} \right )-q_{\tilde{x}} \right )\mathbf{1}\left ( Z_{i}=z,X_{i}=x \right )\left ( \mathbf{1}\left ( D_{i}=1 \right )-p_{D|z,x} \right )}{\mathbb{P}\left ( Z_{i}=z,X_{i}=x \right )} \right ]=0,
\end{equation*}

\begin{equation*}
    \left [ \Sigma_{FS}(z,x,\tilde{z},\tilde{x}) \right ]_{2,1}=\left [ \Sigma_{FS}(z,x,\tilde{z},\tilde{x}) \right ]_{1,2}=0,
\end{equation*}

\begin{equation*}
    \left [ \Sigma_{FS}(z,x,\tilde{z},\tilde{x}) \right ]_{2,2}=\mathbb{E}\left [ \mathbf{1}\left ( X_{i}=x \right )\mathbf{1}\left ( X_{i}=\tilde{x} \right )\left ( \mathbf{1}\left ( Z_{i}=z \right )-p_{z|x} \right )\left ( \mathbf{1}\left ( Z_{i}=\tilde{z} \right ) -p_{\tilde{z}|\tilde{x}}\right ) \right ],
\end{equation*}

\begin{equation*}
    \left [ \Sigma_{FS}(z,x,\tilde{z},\tilde{x}) \right ]_{2,3}=\mathbb{E}\left [ \frac{\left ( \mathbf{1}\left ( X_{i}=\tilde{x} \right )-q_{\tilde{x}} \right )\left ( \mathbf{1}\left ( Z_{i}=z \right )-p_{z|x} \right )}{q_{x}} \right ]=0,
\end{equation*}

\begin{equation*}
    \left [ \Sigma_{FS}(z,x,\tilde{z},\tilde{x}) \right ]_{3,1}=\left [ \Sigma_{FS}(z,x,\tilde{z},\tilde{x}) \right ]_{1,3}=0,
\end{equation*}

\begin{equation*}
    \left [ \Sigma_{FS}(z,x,\tilde{z},\tilde{x}) \right ]_{3,2}=\left [ \Sigma_{FS}(z,x,\tilde{z},\tilde{x}) \right ]_{2,3}=0,
\end{equation*}

\begin{equation*}
    \left [ \Sigma_{FS}(z,x,\tilde{z},\tilde{x}) \right ]_{3,3}=\mathbb{E}\left [ \left ( \mathbf{1}\left ( X_{i}=\tilde{x} \right )-q_{\tilde{x}} \right )\left ( \mathbf{1}\left ( X_{i}=x \right )-q_{x} \right ) \right ]
\end{equation*}

Repeating the procedure for the estimators used in the plug-in of the reduced form yields the asymptotic distribution $\textbf{Z}_{Y}(y,z,x)$, which concludes the proof
\end{proof}

The following lemma provides the asymptotic distributions for the bounds of potential quantities.

\begin{lemma}
    Suppose Assumptions 1-3 and 6-8 hold. Then, 

    \begin{equation*}
        \sqrt{N}\begin{pmatrix}
\widehat{\tau}^{LB}_{D(z)}(c_{1},x)-\tau^{LB}_{D(z)}(c_{1},x) \\\widehat{\tau}^{UB}_{D(z)}(c_{1},x)-\tau^{UB}_{D(z)}(c_{1},x)
\end{pmatrix}\overset{d}{\rightarrow}\tilde{\textbf{Z}}_{D}(d,z,x,c_{1})
    \end{equation*}

    and

    \begin{equation*}
        \sqrt{N}\begin{pmatrix}
\widehat{\tau}^{LB}_{Y(D(z))}(c_{2},x)-\tau^{LB}_{Y(D(z))}(c_{2},x) \\\widehat{\tau}^{UB}_{Y(D(z))}(c_{2},x)-\tau^{UB}_{Y(D(z))}(c_{2},x)
\end{pmatrix}\overset{d}{\rightarrow}\tilde{\textbf{Z}}_{Y}(y,z,x,c_{2})
    \end{equation*}

    both tight elements of $l^{\infty}\left ( \left\{ 0,1\right\}\times\left\{ 0,1\right\}\times\mathcal{S}(X),\mathbb{R}^{2} \right )$.
\end{lemma}

\begin{proof}
    Let $\theta_{0}=\left ( p_{D|z,x},p_{z|x},q_{x} \right )$ and $\widehat{\theta}=\left ( \widehat{p}_{D|z,x},\widehat{p}_{z|x},\widehat{q}_{x} \right )$. For a fixed $d$ and fixed $c_{1}$, define the mapping

\begin{align*}
    & \phi_{D}:l^{\infty}\left ( \left\{ 0,1\right\}\times\left\{ 0,1\right\}\times \mathcal{S}(X) \right )\times l^{\infty}\left ( \left\{ 0,1\right\}\times \mathcal{S}(X) \right )\times l^{\infty}\left ( \mathcal{S}(X) \right )\\&\rightarrow l^{\infty}\left ( \left\{ 0,1\right\},\mathcal{S}(X),\mathbb{R}^{2} \right )
\end{align*}

by

\begin{equation*}
    \left [ \phi_{D}(\theta) \right ](z,x)=\begin{pmatrix}
\min\left\{ \frac{\theta^{(1)}(d,z,x)\theta^{(2)}(z,x)}{\theta^{(2)}(z,x)-c_{1}},\frac{\theta^{(1)}(d,z,x)\theta^{(2)}(z,x)+c_{1}}{\theta^{(2)}(z,x)+c_{1}},\theta^{(1)}(d,z,x)\theta^{(2)}(z,x)+(1-\theta^{(2)}(z,x))\right\} \\
\max\left\{ \frac{\theta^{(1)}(d,z,x)\theta^{(2)}(z,x)}{\theta^{(2)}(z,x)+c_{1}},\frac{\theta^{(1)}(d,z,x)\theta^{(2)}(z,x)-c_{1}}{\theta^{(2)}(z,x)-c_{1}},\theta^{(1)}(d,z,x)\theta^{(2)}(z,x)\right\}
\end{pmatrix}
\end{equation*}

where $\theta^{(j)}$ is the j-th component of $\theta$. Note that

\begin{equation*}
    \left [ \phi_{D}(\theta_{0}) \right ](z,x)=\begin{pmatrix}
\tau^{UB}_{D(z)}(c_{1},x) \\
\tau^{LB}_{D(z)}(c_{1},x)
\end{pmatrix}
\end{equation*}

The mapping $\phi_{D}$ is comprised with max and min operators, along with six other functions. We begin by computing the Hadamard derivative of these functions with respect to $\theta$ using \cite{fangsantos} and the Chain rule for Hadamard differentiable functions to obtain the derivative of $\phi_{D}$.

Let $h\in\mathbb{R}^{2}$. First, consider $\left [ \delta_{D,1}(\theta) \right ](z,x)=\frac{\theta^{(1)}(d,z,x)\theta^{(2)}(z,x)}{\theta^{2}(z,x)+c_{1}}$, which has Hadamard derivative equal to

\begin{equation*}
    \left [ \delta_{D,1,\theta}^{'}(h) \right ](z,x)=\frac{\theta^{(1)}(d,z,x)h^{(2)}(z,x)+h^{(1)}(d,z,x)\theta^{(2)}(z,x)}{\theta^{(2)}(z,x)+c_{1}}-\frac{\theta^{(1)}(d,z,x)\theta^{(2)}(z,x)h^{(2)}(z,x)}{\left ( \theta^{(2)}(z,x)+c_{1} \right )^{2}}
\end{equation*}

Next, $\left [ \delta_{D,2}(\theta) \right ](z,x)=\frac{\theta^{(1)}(d,z,x)\theta^{(2)}(z,x)-c_{1}}{\theta^{(2)}(z,x)-c_{1}}$ has Hadamard derivative equal to

\begin{equation*}
    \left [ \delta_{D,2,\theta}^{'}(h) \right ](z,x)=\frac{\theta^{(1)}(d,z,x)h^{(2)}(z,x)+h^{(1)}(d,z,x)\theta^{(2)}(z,x)}{\theta^{(2)}(z,x)-c_{1}}-\frac{\left ( \theta^{(1)}(d,z,x)\theta^{(2)}(z,x)-c_{1} \right )h^{(2)}(z,x)}{\left ( \theta^{(2)}(z,x)-c_{1} \right )^{2}}
\end{equation*}

Next, $\left [ \delta_{D,3}(\theta) \right ](z,x)=\theta^{(1)}(d,z,x)\theta^{2}(z,x)$ has Hadamard derivative equal to

\begin{equation*}
    \left [ \delta_{D,3,\theta}^{'}(h) \right ](z,x)=h^{(1)}(d,z,x)\theta^{(2)}(z,x)+\theta^{(1)}(d,z,x)+h^{(2)}(z,x)
\end{equation*}

Now, we turn to the functionals inside the $\min$ operator. First, we have $\left [ \delta_{D,4}(\theta) \right ](z,x)=\frac{\theta^{(1)}(d,z,x)\theta^{(2)}(z,x)}{\theta^{(2)}(z,x)-c_{1}}$, which has Hadamard derivative equal to

\begin{equation*}
    \left [ \delta_{D,4,\theta}^{'}(h) \right ](z,x)=\frac{\theta^{(1)}(d,z,x)h^{(2)}(z,x)+h^{(1)}(d,z,x)\theta^{(2)}(z,x)}{\theta^{(2)}(z,x)-c_{1}}-\frac{\theta^{(1)}(d,z,x)\theta^{(2)}(z,x)h^{(2)}(z,x)}{\left ( \theta^{(2)}(z,x)-c_{1} \right )^{2}}
\end{equation*}

Next, $\left [ \delta_{D,5}(\theta) \right ](z,x)=\frac{\theta^{(1)}(d,z,x)\theta^{(2)}(z,x)+c_{1}}{\theta^{(2)}(z,x)+c_{1}}$ has Hadamard derivative equal to

\begin{equation*}
    \left [ \delta_{D,5,\theta}^{'}(h) \right ](z,x)=\frac{\theta^{(1)}(d,z,x)h^{(2)}(z,x)+h^{(1)}(d,z,x)\theta^{(2)}(z,x)}{\theta^{(2)}(z,x)+c_{1}}-\frac{\left ( \theta^{(1)}(d,z,x)\theta^{(2)}(z,x)+c_{1} \right )h^{(2)}(z,x)}{\left ( \theta^{(2)}(z,x)+c_{1} \right )^{2}}
\end{equation*}

Finally, $\left [ \delta_{D,6,\theta}^{'}(h) \right ](z,x)=h^{(1)}(d,z,x)\theta^{(2)}(z,x)+h^{(2)}(z,x)(\theta^{(1)}(d,z,x)-1)$.

Using this notation, we write the functional $\phi_{D}$ as

\begin{equation*}
    \phi_{D}(\theta)=\begin{pmatrix}
\min\left\{ \delta_{D,4}(\theta),\delta_{D,5}(\theta),\delta_{D,6}(\theta\right\} \\
\max\left\{ \delta_{D,1}(\theta),\delta_{D,2}(\theta),\delta_{D,3}(\theta)\right\}
\end{pmatrix}
\end{equation*}

Using the chain rule \citep{mastenpoirier2020}, the Hadamard derivative of $\phi_{D}$ at $\theta_{0}$ is

\begin{equation*}
    \phi_{D,\theta_{0}}^{'}(h)=\begin{pmatrix}
\mathbf{1}\left ( \delta_{D,6}(\theta_{0})>\max\left\{\delta_{D,4}(\theta_{0}),\delta_{D,5}(\theta_{0}) \right\} \right )\delta_{D,6,\theta_{0}}^{'}(h)\\
+\mathbf{1}\left ( \delta_{D,5}(\theta_{0})>\max\left\{\delta_{D,4}(\theta_{0}),\delta_{D,6}(\theta_{0}) \right\} \right )\delta_{D,5,\theta_{0}}^{'}(h) \\
+\mathbf{1}\left ( \delta_{D,4}(\theta_{0})>\max\left\{\delta_{D,5}(\theta_{0}),\delta_{D,6}(\theta_{0}) \right\} \right )\delta_{D,4,\theta_{0}}^{'}(h)\\
+\mathbf{1}\left ( \delta_{D,6}(\theta_{0})=\delta_{D,5}(\theta_{0})>\delta_{D,4}(\theta_{0})  \right )\min\left\{ \delta_{D,6,\theta_{0}}^{'}(h),\delta_{D,5,\theta_{0}}^{'}(h)\right\}\\
+\mathbf{1}\left ( \delta_{D,6}(\theta_{0})=\delta_{D,4}(\theta_{0})>\delta_{D,5}(\theta_{0})  \right )\min\left\{ \delta_{D,6,\theta_{0}}^{'}(h),\delta_{D,4,\theta_{0}}^{'}(h)\right\}\\
+\mathbf{1}\left ( \delta_{D,4}(\theta_{0})=\delta_{D,5}(\theta_{0})>\delta_{D,6}(\theta_{0})  \right )\min\left\{ \delta_{D,4,\theta_{0}}^{'}(h),\delta_{D,5,\theta_{0}}^{'}(h)\right\}\\
+\mathbf{1}\left ( \delta_{D,6}(\theta_{0})=\delta_{D,5}(\theta_{0})=\delta_{D,4}(\theta_{0})  \right )\min\left\{ \delta_{D,6,\theta_{0}}^{'}(h),\delta_{D,5,\theta_{0}}^{'}(h),\delta_{D,6,\theta_{0}}^{'}(h)\right\}\\
 \\
\mathbf{1}\left ( \delta_{D,3}(\theta_{0})<\min\left\{\delta_{D,1}(\theta_{0}),\delta_{D,2}(\theta_{0}) \right\} \right )\delta_{D,3,\theta_{0}}^{'}(h)\\
+\mathbf{1}\left ( \delta_{D,2}(\theta_{0})<\min\left\{\delta_{D,1}(\theta_{0}),\delta_{D,3}(\theta_{0}) \right\} \right )\delta_{D,2,\theta_{0}}^{'}(h) \\
+\mathbf{1}\left ( \delta_{D,3}(\theta_{0})<\min\left\{\delta_{D,1}(\theta_{0}),\delta_{D,2}(\theta_{0}) \right\} \right )\delta_{D,3,\theta_{0}}^{'}(h)\\
+\mathbf{1}\left ( \delta_{D,3}(\theta_{0})=\delta_{D,2}(\theta_{0})<\delta_{D,1}(\theta_{0})  \right )\max\left\{ \delta_{D,3,\theta_{0}}^{'}(h),\delta_{D,2,\theta_{0}}^{'}(h)\right\}\\
+\mathbf{1}\left ( \delta_{D,3}(\theta_{0})=\delta_{D,1}(\theta_{0})<\delta_{D,2}(\theta_{0})  \right )\max\left\{ \delta_{D,3,\theta_{0}}^{'}(h),\delta_{D,1,\theta_{0}}^{'}(h)\right\}\\
+\mathbf{1}\left ( \delta_{D,1}(\theta_{0})=\delta_{D,2}(\theta_{0})>\delta_{D,3}(\theta_{0})  \right )\max\left\{ \delta_{D,1,\theta_{0}}^{'}(h),\delta_{D,2,\theta_{0}}^{'}(h)\right\}\\
+\mathbf{1}\left ( \delta_{D,3}(\theta_{0})=\delta_{D,2}(\theta_{0})=\delta_{D,1}(\theta_{0})  \right )\max\left\{ \delta_{D,3,\theta_{0}}^{'}(h),\delta_{D,2,\theta_{0}}^{'}(h),\delta_{D,1,\theta_{0}}^{'}(h)\right\}\\
\end{pmatrix}
\end{equation*}

By Lemma 3, $\sqrt{N}\left ( \widehat{\theta}-\theta_{0} \right )\overset{d}{\rightarrow}\textbf{Z}_{D}(d,x)$. Using the Delta Method for Hadamard differentiable functions, we obtain

\begin{equation*}
    \left [ \sqrt{N}\left ( \phi_{D}(\widehat{\theta})-\phi_{D}(\theta_{0}) \right ) \right ](z,x)\overset{d}{\rightarrow}\left [ \phi_{D,\theta_{0}}^{'}(\textbf{Z}_{D}) \right ](z,x)\equiv \tilde{\textbf{Z}}_{D}(z,x)
\end{equation*}

which concludes the proof for the bounds of potential treatments. Now, consider the bounds for potential outcomes.

Let $\theta_{0}=\left ( p_{Y|z,x},p_{z|x},q_{x} \right )$ and $\widehat{\theta}=\left ( \widehat{p}_{Y|z,x},\widehat{p}_{z|x},\widehat{q}_{x} \right )$. For a fixed $y$ and fixed $c_{2}$, define the mapping

\begin{align*}
    & \phi_{Y}:l^{\infty}\left ( \left\{ 0,1\right\}\times\left\{ 0,1\right\}\times \mathcal{S}(X) \right )\times l^{\infty}\left ( \left\{ 0,1\right\}\times \mathcal{S}(X) \right )\times l^{\infty}\left ( \mathcal{S}(X) \right )\\&\rightarrow l^{\infty}\left ( \left\{ 0,1\right\},\mathcal{S}(X),\mathbb{R}^{2} \right )
\end{align*}

by

\begin{equation*}
    \left [ \phi_{Y}(\theta) \right ](z,x)=\begin{pmatrix}
\min\left\{ \frac{\theta^{(1)}(y,z,x)\theta^{(2)}(z,x)}{\theta^{(2)}(z,x)-c_{1}},\frac{\theta^{(1)}(y,z,x)\theta^{(2)}(z,x)+c_{1}}{\theta^{(2)}(z,x)+c_{1}},\theta^{(1)}(y,z,x)\theta^{(2)}(z,x)+(1-\theta^{(2)}(z,x))\right\} \\
\max\left\{ \frac{\theta^{(1)}(y,z,x)\theta^{(2)}(z,x)}{\theta^{(2)}(z,x)+c_{1}},\frac{\theta^{(1)}(y,z,x)\theta^{(2)}(z,x)-c_{1}}{\theta^{(2)}(z,x)-c_{1}},\theta^{(1)}(y,z,x)\theta^{(2)}(z,x)\right\}
\end{pmatrix}
\end{equation*}

where $\theta^{(j)}$ is the j-th component of $\theta$. Note that

\begin{equation*}
    \left [ \phi_{Y}(\theta_{0}) \right ](z,x)=\begin{pmatrix}
\tau^{UB}_{Y(D(z))}(c_{2},x) \\
\tau^{LB}_{Y(D(z))}(c_{2},x)
\end{pmatrix}
\end{equation*}

The mapping $\phi_{Y}$ is also comprised with max and min operators, along with six other functions.

We begin by computing the Hadamard derivative of these functions with respect to $\theta$ using Fang and Santos (2019) and use the Chain rule for Hadamard differentiable functions to obtain the derivative of $\phi_{Y}$.

First, consider $\left [ \delta_{Y,1}(\theta) \right ](z,x)=\frac{\theta^{(1)}(y,z,x)\theta^{(2)}(z,x)}{\theta^{2}(z,x)+c_{2}}$, which has Hadamard derivative equal to

\begin{equation*}
    \left [ \delta_{Y,1,\theta}^{'}(h) \right ](z,x)=\frac{\theta^{(1)}(y,z,x)h^{(2)}(z,x)+h^{(1)}(y,z,x)\theta^{(2)}(z,x)}{\theta^{(2)}(z,x)+c_{2}}-\frac{\theta^{(1)}(y,z,x)\theta^{(2)}(z,x)h^{(2)}(z,x)}{\left ( \theta^{(2)}(z,x)+c_{2} \right )^{2}}
\end{equation*}

Next, $\left [ \delta_{Y,2}(\theta) \right ](z,x)=\frac{\theta^{(1)}(y,z,x)\theta^{(2)}(z,x)-c_{2}}{\theta^{(2)}(z,x)-c_{2}}$ has Hadamard derivative equal to

\begin{equation*}
    \left [ \delta_{Y,2,\theta}^{'}(h) \right ](z,x)=\frac{\theta^{(1)}(y,z,x)h^{(2)}(z,x)+h^{(1)}(y,z,x)\theta^{(2)}(z,x)}{\theta^{(2)}(z,x)-c_{2}}-\frac{\left ( \theta^{(1)}(y,z,x)\theta^{(2)}(z,x)-c_{2} \right )h^{(2)}(z,x)}{\left ( \theta^{(2)}(z,x)-c_{2} \right )^{2}}
\end{equation*}

Next, $\left [ \delta_{Y,3}(\theta) \right ](z,x)=\theta^{(1)}(y,z,x)\theta^{2}(z,x)$ has Hadamard derivative equal to

\begin{equation*}
    \left [ \delta_{Y,3,\theta}^{'}(h) \right ](z,x)=h^{(1)}(y,z,x)\theta^{(2)}(z,x)+\theta^{(1)}(y,z,x)+h^{(2)}(z,x)
\end{equation*}

Now, we turn to the functionals inside the $\min$ operator. First, we have $\left [ \delta_{Y,4}(\theta) \right ](z,x)=\frac{\theta^{(1)}(y,z,x)\theta^{(2)}(z,x)}{\theta^{(2)}(z,x)-c_{2}}$, which has Hadamard derivative equal to

\begin{equation*}
    \left [ \delta_{Y,4,\theta}^{'}(h) \right ](z,x)=\frac{\theta^{(1)}(y,z,x)h^{(2)}(z,x)+h^{(1)}(y,z,x)\theta^{(2)}(z,x)}{\theta^{(2)}(z,x)-c_{2}}-\frac{\theta^{(1)}(y,z,x)\theta^{(2)}(z,x)h^{(2)}(z,x)}{\left ( \theta^{(2)}(z,x)-c_{2} \right )^{2}}
\end{equation*}

Next, $\left [ \delta_{Y,5}(\theta) \right ](z,x)=\frac{\theta^{(1)}(y,z,x)\theta^{(2)}(z,x)+c_{2}}{\theta^{(2)}(z,x)+c_{2}}$ has Hadamard derivative equal to

\begin{equation*}
    \left [ \delta_{Y,5,\theta}^{'}(h) \right ](z,x)=\frac{\theta^{(1)}(y,z,x)h^{(2)}(z,x)+h^{(1)}(y,z,x)\theta^{(2)}(z,x)}{\theta^{(2)}(z,x)+c_{2}}-\frac{\left ( \theta^{(1)}(y,z,x)\theta^{(2)}(z,x)+c_{2} \right )h^{(2)}(z,x)}{\left ( \theta^{(2)}(z,x)+c_{2} \right )^{2}}
\end{equation*}

Finally, $\left [ \delta_{Y,6,\theta}(h) \right ](z,x)=h^{(1)}(y,z,x)\theta^{(2)}(z,x)+h^{(2)}(z,x)(\theta^{(1)}(y,z,x)-1)$

\begin{equation*}
    \phi_{Y}(\theta)=\begin{pmatrix}
\min\left\{ \delta_{Y,4}(\theta),\delta_{Y,5}(\theta),\delta_{Y,6}(\theta)\right\} \\
\max\left\{ \delta_{Y,1}(\theta),\delta_{Y,2}(\theta),\delta_{Y,3}(\theta)\right\}
\end{pmatrix}
\end{equation*}

Using the chain rule from \cite{mastenpoirier2020}, the Hadamard derivative of $\phi_{Y}$ at $\theta_{0}$ is

\begin{equation*}
    \phi_{Y,\theta_{0}}^{'}(h)=\begin{pmatrix}
\mathbf{1}\left ( \delta_{Y,6}(\theta_{0})>\max\left\{\delta_{Y,4}(\theta_{0}),\delta_{Y,5}(\theta_{0}) \right\} \right )\delta_{Y,6,\theta_{0}}^{'}(h)\\
+\mathbf{1}\left ( \delta_{Y,5}(\theta_{0})>\max\left\{\delta_{Y,4}(\theta_{0}),\delta_{Y,6}(\theta_{0}) \right\} \right )\delta_{Y,5,\theta_{0}}^{'}(h) \\
+\mathbf{1}\left ( \delta_{Y,4}(\theta_{0})>\max\left\{\delta_{Y,5}(\theta_{0}),\delta_{Y,6}(\theta_{0}) \right\} \right )\delta_{Y,4,\theta_{0}}^{'}(h)\\
+\mathbf{1}\left ( \delta_{Y,6}(\theta_{0})=\delta_{Y,5}(\theta_{0})>\delta_{Y,4}(\theta_{0})  \right )\min\left\{ \delta_{Y,6,\theta_{0}}^{'}(h),\delta_{Y,5,\theta_{0}}^{'}(h)\right\}\\
+\mathbf{1}\left ( \delta_{Y,6}(\theta_{0})=\delta_{Y,4}(\theta_{0})>\delta_{Y,5}(\theta_{0})  \right )\min\left\{ \delta_{Y,6,\theta_{0}}^{'}(h),\delta_{Y,4,\theta_{0}}^{'}(h)\right\}\\
+\mathbf{1}\left ( \delta_{Y,4}(\theta_{0})=\delta_{Y,5}(\theta_{0})>\delta_{Y,6}(\theta_{0})  \right )\min\left\{ \delta_{Y,4,\theta_{0}}^{'}(h),\delta_{Y,5,\theta_{0}}^{'}(h)\right\}\\
+\mathbf{1}\left ( \delta_{Y,6}(\theta_{0})=\delta_{Y,5}(\theta_{0})=\delta_{Y,4}(\theta_{0})  \right )\min\left\{ \delta_{Y,6,\theta_{0}}^{'}(h),\delta_{Y,5,\theta_{0}}^{'}(h),\delta_{Y,6,\theta_{0}}^{'}(h)\right\}\\
 \\
\mathbf{1}\left ( \delta_{Y,3}(\theta_{0})<\min\left\{\delta_{Y,1}(\theta_{0}),\delta_{Y,2}(\theta_{0}) \right\} \right )\delta_{Y,3,\theta_{0}}^{'}(h)\\
+\mathbf{1}\left ( \delta_{Y,2}(\theta_{0})<\min\left\{\delta_{Y,1}(\theta_{0}),\delta_{Y,3}(\theta_{0}) \right\} \right )\delta_{Y,2,\theta_{0}}^{'}(h) \\
+\mathbf{1}\left ( \delta_{Y,3}(\theta_{0})<\min\left\{\delta_{Y,1}(\theta_{0}),\delta_{Y,2}(\theta_{0}) \right\} \right )\delta_{Y,3,\theta_{0}}^{'}(h)\\
+\mathbf{1}\left ( \delta_{Y,3}(\theta_{0})=\delta_{Y,2}(\theta_{0})<\delta_{Y,1}(\theta_{0})  \right )\max\left\{ \delta_{Y,3,\theta_{0}}^{'}(h),\delta_{Y,2,\theta_{0}}^{'}(h)\right\}\\
+\mathbf{1}\left ( \delta_{Y,3}(\theta_{0})=\delta_{Y,1}(\theta_{0})<\delta_{Y,2}(\theta_{0})  \right )\max\left\{ \delta_{Y,3,\theta_{0}}^{'}(h),\delta_{Y,1,\theta_{0}}^{'}(h)\right\}\\
+\mathbf{1}\left ( \delta_{Y,1}(\theta_{0})=\delta_{Y,2}(\theta_{0})>\delta_{Y,3}(\theta_{0})  \right )\max\left\{ \delta_{Y,1,\theta_{0}}^{'}(h),\delta_{Y,2,\theta_{0}}^{'}(h)\right\}\\
+\mathbf{1}\left ( \delta_{Y,3}(\theta_{0})=\delta_{Y,2}(\theta_{0})=\delta_{Y,1}(\theta_{0})  \right )\max\left\{ \delta_{Y,3,\theta_{0}}^{'}(h),\delta_{Y,2,\theta_{0}}^{'}(h),\delta_{Y,1,\theta_{0}}^{'}(h)\right\}\\
\end{pmatrix}
\end{equation*}

By Lemma 3, $\sqrt{N}\left ( \widehat{\theta}-\theta_{0} \right )\overset{d}{\rightarrow}\textbf{Z}_{RF}(d,x)$. Using the Delta Method for Hadamard differentiable functions, we obtain

\begin{equation*}
    \left [ \sqrt{N}\left ( \phi_{Y}(\widehat{\theta})-\phi_{Y}(\theta_{0}) \right ) \right ](z,x)\overset{d}{\rightarrow}\left [ \phi_{Y,\theta_{0}}^{'}(\textbf{Z}_{Y}) \right ](z,x)\equiv \tilde{\textbf{Z}}_{Y}(z,x)
\end{equation*}

which concludes the proof.

\end{proof}

The next lemma provides the asymptotic properties of the nonparametric estimators of conditional probabilities which are used as plug-ins in the estimator of the LATE bounds under joint \textit{$c$-dependence}:

\begin{lemma}
    Suppose Assumptions 6 and 7 hold. Then,

    \begin{equation*}
        \sqrt{N}\begin{pmatrix}
\widehat{p}_{Y|z,x}-p_{Y|z,x} \\
\widehat{p}_{D|z,x}-p_{D|z,x} \\
\widehat{p}_{z|x}-p_{z|x}\\
\widehat{q}_{x}-q_{x}
\end{pmatrix}\rightarrow \textbf{Z}_{j}(y,d,z,x)
    \end{equation*}

    which is a mean-zero Gaussian processes in $l^{\infty}\left(\left\{0,1\right\}^{3}\times\mathcal{S}(X),\mathbb{R}^{3}\right)$ with covariance kernel equal to $\Sigma_{j}=\mathbb{E}\left [ \textbf{Z}_{j}(y,d,z,x)\textbf{Z}_{j}(y,d,\tilde{z},\tilde{x})^{'}  \right ]$
\end{lemma}

\begin{proof}
    Using a Taylor expansion and bracket entropy arguments as in Lemma 3, the following asymptotically linear representations are obtained: 

    \begin{align*}
        &\widehat{p}_{Y|z,x}-p_{Y|z,x}=\frac{1}{N}\frac{\sum_{i=1}^{N}\mathbf{1}\left ( Z_{i}=z,X_{i}=x \right )\left ( \mathbf{1}\left ( Y_{i}=1 \right )-p_{Y|z,x} \right )}{\mathbb{P}\left ( Z=z,X=x \right )}+o_{p}(N^{-1/2})\\&\widehat{p}_{D|z,x}-p_{D|z,x}=\frac{1}{N}\frac{\sum_{i=1}^{N}\mathbf{1}\left ( Z_{i}=z,X_{i}=x \right )\left ( \mathbf{1}\left ( D_{i}=1 \right )-p_{D|z,x} \right )}{\mathbb{P}\left ( Z=z,X=x \right )}+o_{p}(N^{-1/2})\\&\widehat{p}_{z|x}-p_{z|x}=\frac{1}{N}\frac{\sum_{i=1}^{N}\mathbf{1}\left ( X_{i}=x \right )\left ( \mathbf{1}\left ( Z_{i}=z \right )-p_{z|x} \right )}{q_{x}}+o_{p}(N^{-1/2})\\&\widehat{q}_{x}-q_{x}=\frac{1}{N}\sum_{i=1}^{N}\left ( \mathbf{1}\left ( X_{i}=x \right )-q_{x} \right )+o_{p}(N^{-1/2})
    \end{align*}

    The covariance kernel $\Sigma_{j}$ can be calculated as follows:

\begin{equation*}
    \left [ \Sigma_{j}  \right ]_{1,1}=\mathbb{E}\left [ \frac{\mathbf{1}\left ( Z_{i}=z,X_{i}=x \right )\mathbf{1}\left ( Z_{i}=\tilde{z},X_{i}=\tilde{x} \right )\left ( \mathbf{1}\left ( Y_{i}=1 \right )-p_{Y|z,x} \right )\left ( \mathbf{1}\left ( Y_{i}=1 \right )-p_{Y|\tilde{z},\tilde{x}} \right )}{\mathbb{P}\left ( Z=z,X=x \right )\mathbb{P}\left ( Z=\tilde{z},X=\tilde{x} \right )} \right ],
\end{equation*}

\begin{equation*}
    \left [ \Sigma_{j}  \right ]_{1,2}=\mathbb{E}\left [ \frac{\mathbf{1}\left ( Z_{i}=z,X_{i}=x \right )\mathbf{1}\left ( Z_{i}=\tilde{z},X_{i}=\tilde{x} \right )\left ( \mathbf{1}\left ( Y_{i}=1 \right )-p_{Y|z,x} \right )\left ( \mathbf{1}\left ( D_{i}=1 \right )-p_{D|\tilde{z},\tilde{x}} \right )}{\mathbb{P}\left ( Z=z,X=x \right )\mathbb{P}\left ( Z=\tilde{z},X=\tilde{x} \right )} \right ]=0,
\end{equation*}

\begin{equation*}
    \left [ \Sigma_{j}  \right ]_{1,3}=\mathbb{E}\left [ \frac{\mathbf{1}\left ( Z_{i}=z,X_{i}=\tilde{x} \right )\left ( \mathbf{1}\left ( Y_{i}=1 \right )-p_{Y|z,x} \right )\left ( \mathbf{1}\left ( Z_{i}=\tilde{z} \right )-p_{\tilde{z}|\tilde{x}} \right )}{p_{z|x}q_{x}q_{\tilde{x}}} \right ]=0,
\end{equation*}

\begin{equation*}
    \left [ \Sigma_{j}  \right ]_{1,4}=\mathbb{E}\left [ \frac{\left ( \mathbf{1}\left ( X_{i}=\tilde{x} \right )-q_{\tilde{x}} \right )\mathbf{1}\left ( Z_{i}=z,X_{i}=x \right )\left ( \mathbf{1}\left ( Y_{i}=1 \right )-p_{Y|z,x} \right )}{\mathbb{P}\left ( Z=z,X=x \right )} \right ]=0,
\end{equation*}

\begin{equation*}
    \left [ \Sigma_{j} \right ]_{2,1}=\left [ \Sigma_{j} \right ]_{1,2}=0,
\end{equation*}

\begin{equation*}
    \left [ \Sigma_{j}  \right ]_{2,2}=\mathbb{E}\left [ \frac{\mathbf{1}\left ( Z_{i}=z,X_{i}=x \right )\mathbf{1}\left ( Z_{i}=\tilde{z},X_{i}=\tilde{x} \right )\left ( \mathbf{1}\left ( D_{i}=1 \right )-p_{D|z,x} \right )\left ( \mathbf{1}\left ( D_{i}=1 \right )-p_{D|\tilde{z},\tilde{x}} \right )}{\mathbb{P}\left ( Z=z,X=x \right )\mathbb{P}\left ( Z=\tilde{z},X=\tilde{x} \right )} \right ],
\end{equation*}

\begin{equation*}
    \left [ \Sigma_{j}  \right ]_{2,3}=\mathbb{E}\left [ \frac{\mathbf{1}\left ( Z_{i}=z,X_{i}=\tilde{x} \right )\left ( \mathbf{1}\left ( D_{i}=1 \right )-p_{D|z,x} \right )\left ( \mathbf{1}\left ( Z_{i}=\tilde{z} \right )-p_{\tilde{z}|\tilde{x}} \right )}{p_{z|x}q_{x}q_{\tilde{x}}} \right ]=0,
\end{equation*}

\begin{equation*}
    \left [ \Sigma_{j}  \right ]_{2,4}=\mathbb{E}\left [ \frac{\left ( \mathbf{1}\left ( X_{i}=\tilde{x} \right )-q_{\tilde{x}} \right )\mathbf{1}\left ( Z_{i}=z,X_{i}=x \right )\left ( \mathbf{1}\left ( D_{i}=1 \right )-p_{D|z,x} \right )}{\mathbb{P}\left ( Z=z,X=x \right )} \right ]=0,
\end{equation*}

\begin{equation*}
    \left [ \Sigma_{j} \right ]_{3,1}=\left [ \Sigma_{j} \right ]_{1,3}=0,
\end{equation*}

\begin{equation*}
    \left [ \Sigma_{j} \right ]_{3,2}=\left [ \Sigma_{j} \right ]_{2,3}=0,
\end{equation*}

\begin{equation*}
    \left [ \Sigma_{j} \right ]_{3,3}=\mathbb{E}\left [ \frac{\mathbf{1}\left ( X_{i}=x \right )\mathbf{1}\left ( X_{i}=\tilde{x} \right )\left ( \mathbf{1}\left ( Z_{i}=z \right )-p_{z|x} \right )\left ( \mathbf{1}\left ( Z_{i}=\tilde{z} \right )-p_{\tilde{z}|\tilde{x}} \right )}{q_{x}q_{\tilde{x}}} \right ],
\end{equation*}

\begin{equation*}
    \left [ \Sigma_{j} \right ]_{3,4}=\mathbb{E}\left [ \frac{\left ( \mathbf{1}\left ( X_{i}=\tilde{x} \right )-q_{\tilde{x}} \right )\left ( \mathbf{1}\left ( Z_{i}=z \right )-p_{z|x} \right )}{q_{x}} \right ]=0,
\end{equation*}

\begin{equation*}
    \left [ \Sigma_{j} \right ]_{4,1}=\left [ \Sigma_{j} \right ]_{1,4}=0,
\end{equation*}

\begin{equation*}
    \left [ \Sigma_{j} \right ]_{4,2}=\left [ \Sigma_{j} \right ]_{2,4}=0,
\end{equation*}

\begin{equation*}
    \left [ \Sigma_{j} \right ]_{4,3}=\left [ \Sigma_{j} \right ]_{3,4}=0,
\end{equation*}

\begin{equation*}
    \left [ \Sigma_{j} \right ]_{4,4}=\mathbb{E}\left [\left ( \mathbf{1}\left ( X_{i}=x \right )-q_{x} \right )\left ( \mathbf{1}\left ( X_{i}=\tilde{x} \right )-q_{\tilde{x}} \right )  \right ]
\end{equation*}

\end{proof}

The next lemma provides the asymptotic properties for the estimators of the bounds of potential quantities under joint \textit{$c$-dependence}:

\begin{lemma}
    Suppose Assumptions 1-3 and 6-9 hold. Then, 

    \begin{equation*}
        \sqrt{N}\begin{pmatrix}
\widehat{\tau}^{LB}_{D(z)}(c,x)-\tau^{LB}_{D(z)}(c,x) \\\widehat{\tau}^{UB}_{D(z)}(c,x)-\tau^{UB}_{D(z)}(c,x)
\end{pmatrix}\overset{d}{\rightarrow}\tilde{\textbf{Z}}_{j,D}(d,z,x,c)
    \end{equation*}

    and

    \begin{equation*}
        \sqrt{N}\begin{pmatrix}
\widehat{\tau}^{LB}_{Y(D(z))}(c,x)-\tau^{LB}_{Y(D(z))}(c,x) \\\widehat{\tau}^{UB}_{Y(D(z))}(c,x)-\tau^{UB}_{Y(D(z))}(c,x)
\end{pmatrix}\overset{d}{\rightarrow}\tilde{\textbf{Z}}_{j,Y}(y,z,x,c)
    \end{equation*}

    both tight elements of $l^{\infty}\left ( \left\{ 0,1\right\}\times\left\{ 0,1\right\}\times\mathcal{S}(X),\mathbb{R}^{2} \right )$.
\end{lemma}

\begin{proof}
    Let $\Omega_{0}=\left ( p_{Y|z,x},p_{D|z,x},p_{z|x},q_{x} \right )$ and $\widehat{\Omega}=\left ( \widehat{p}_{Y|z,x},\widehat{p}_{D|z,x},\widehat{p}_{z|x},\widehat{q}_{x} \right )$. For fixed $d$ and $c$, define the mapping

    \begin{align*}
        &\psi_{FS}:l^{\infty}\left ( \left\{ 0,1\right\}^{2} \times\mathcal{S}(X)\right )\times l^{\infty}\left ( \left\{ 0,1\right\}^{2}\times\mathcal{S}(X) \right )\times l^{\infty}\left ( \mathcal{S}(X) \right )\\&\rightarrow l^{\infty}\left ( \left\{ 0,1\right\}\times\mathcal{S}(X),\mathbb{R}^{2} \right )
    \end{align*}

    by

    \begin{equation*}
        \left [ \psi_{FS}(\Omega) \right ](z,x)=\begin{pmatrix}
\min\left\{ \delta_{FS,4}(\Omega),\delta_{FS,5}(\Omega),\delta_{FS,6}(\Omega) \right\} \\\max\left\{ \delta_{FS,1}(\Omega),\delta_{FS,2}(\Omega),\delta_{FS,3}(\Omega)\right\}
\end{pmatrix}
    \end{equation*}

    where 

    \begin{equation*}
        \left [ \delta_{FS,1}(\Omega) \right ](z,x)=\frac{\Omega^{(2)}(d,z,x)\Omega^{(3)}(z,x)}{\Omega^{(3)}(z,x)+c}
    \end{equation*}

    which has Hadamard directional derivative equal to

    \begin{equation*}
        \left [ \delta^{'}_{FS,1,\Omega}(h) \right ](z,x)=\frac{\Omega^{(2)}(d,z,x)h^{(3)}(z,x)+h^{(2)}(d,z,x)\Omega^{(3)}(z,x)}{\Omega^{(3)}(z,x)+c}-\frac{\Omega^{(2)}(d,z,x)\Omega^{(3)}(z,x)h^{(3)}(z,x)}{\left ( \Omega^{(3)}(z,x)+c \right )^{2}}
    \end{equation*}

    \begin{equation*}
        \left [ \delta_{FS,2}(\Omega) \right ](z,x)=\frac{\Omega^{(2)}(d,z,x)\Omega^{(3)}(z,x)-2c}{\Omega^{(3)}(z,x)-c}
    \end{equation*}

    which has Hadamard directional derivative equal to

    \begin{equation*}
        \left [ \delta^{'}_{FS,2,\Omega}(h) \right ](z,x)=\frac{\Omega^{(2)}(d,z,x)h^{(3)}(z,x)+h^{(2)}(d,z,x)\Omega^{(3)}(z,x)}{\Omega^{(3)}(z,x)-c}-\frac{\left ( \Omega^{(2)}(d,z,x)\Omega^{(3)}(z,x)-2c \right )h^{(3)}(z,x)}{\left ( \Omega^{(3)}(z,x)-c \right )^{2}},
    \end{equation*}

    \begin{equation*}
        \left [ \delta_{FS,3}(\Omega) \right ](z,x)=\Omega^{(2)}(d,z,x)\Omega^{(3)}(z,x)
    \end{equation*}

    which has Hadamard directional derivative equal to

    \begin{equation*}
        \left [ \delta^{'}_{FS,3,\Omega}(h) \right ](z,x)=h^{(2)}(d,z,x)\Omega^{(3)}(z,x)+\Omega^{(2)}(d,z,x)h^{(3)}(z,x),
    \end{equation*}

    \begin{equation*}
        \left [ \delta_{FS,4}(\Omega) \right ](z,x)=\frac{\Omega^{(2)}(d,z,x)\Omega^{(3)}(z,x)}{\Omega^{(3)}(z,x)-c}
    \end{equation*}

    which has Hadamard directional derivative equal to

    \begin{equation*}
        \left [ \delta^{'}_{FS,4,\Omega}(h) \right ](z,x)=\frac{\Omega^{(2)}(d,z,x)h^{(3)}(z,x)+h^{(2)}(d,z,x)\Omega^{(3)}(z,x)}{\Omega^{(3)}(z,x)-c}-\frac{\Omega^{(2)}(d,z,x)\Omega^{(3)}(z,x)h^{(2)}(z,x)}{\left ( \Omega^{(3)}(z,x)-c \right )^{2}}
    \end{equation*},

\begin{equation*}
    \left [ \delta_{FS,5}(\Omega) \right ](z,x)=\frac{\Omega^{(1)}(d,z,x)\Omega^{(2)}(z,x)+2c}{\Omega^{(2)}(z,x)+c},
\end{equation*}

which has Hadamard derivative equal to

\begin{equation*}
    \left [ \delta_{FS,5,\Omega}^{'}(h) \right ](z,x)=\frac{\Omega^{(2)}(d,z,x)h^{(3)}(z,x)+h^{(2)}(d,z,x)\Omega^{(3)}(z,x)}{\Omega^{(3)}(z,x)+c}-\frac{\left ( \Omega^{(2)}(d,z,x)\Omega^{(3)}(z,x)+2c \right )h^{(3)}(z,x)}{\left ( \Omega^{(3)}(z,x)+c \right )^{2}},
\end{equation*}

and

\begin{equation*}
    \left [ \delta_{FS,6,\Omega}^{'}(h) \right ](z,x)=h^{(2)}(d,z,x)\Omega^{(3)}(z,x)+h^{(3)}(z,x)(\Omega^{(2)}(d,z,x)-1)
\end{equation*}

Hence, the Hadamard directional derivative of the map $\psi_{FS}$ evaluated at $\Omega_{0}$ is 

\begin{equation*}
    \psi_{FS,\Omega_{0}}^{'}(h)=\begin{pmatrix}
\mathbf{1}\left ( \delta_{FS,6}(\Omega_{0})>\max\left\{\delta_{FS,4}(\Omega_{0}),\delta_{FS,5}(\Omega_{0}) \right\} \right )\delta_{FS,6,\Omega_{0}}^{'}(h)\\
+\mathbf{1}\left ( \delta_{FS,5}(\Omega_{0})>\max\left\{\delta_{FS,4}(\Omega_{0}),\delta_{FS,6}(\Omega_{0}) \right\} \right )\delta_{FS,5,\Omega_{0}}^{'}(h) \\
+\mathbf{1}\left ( \delta_{FS,4}(\Omega_{0})>\max\left\{\delta_{FS,5}(\Omega_{0}),\delta_{FS,6}(\Omega_{0}) \right\} \right )\delta_{FS,4,\Omega_{0}}^{'}(h)\\
+\mathbf{1}\left ( \delta_{FS,6}(\Omega_{0})=\delta_{FS,5}(\Omega_{0})>\delta_{FS,4}(\Omega_{0})  \right )\min\left\{ \delta_{FS,6,\Omega_{0}}^{'}(h),\delta_{FS,5,\Omega_{0}}^{'}(h)\right\}\\
+\mathbf{1}\left ( \delta_{FS,6}(\Omega_{0})=\delta_{FS,4}(\Omega_{0})>\delta_{FS,5}(\Omega_{0})  \right )\min\left\{ \delta_{FS,6,\Omega_{0}}^{'}(h),\delta_{FS,4,\Omega_{0}}^{'}(h)\right\}\\
+\mathbf{1}\left ( \delta_{FS,4}(\Omega_{0})=\delta_{FS,5}(\Omega_{0})>\delta_{FS,6}(\Omega_{0})  \right )\min\left\{ \delta_{FS,4,\Omega_{0}}^{'}(h),\delta_{FS,5,\Omega_{0}}^{'}(h)\right\}\\
+\mathbf{1}\left ( \delta_{FS,6}(\Omega_{0})=\delta_{FS,5}(\Omega_{0})=\delta_{FS,4}(\Omega_{0})  \right )\min\left\{ \delta_{FS,6,\Omega_{0}}^{'}(h),\delta_{FS,5,\Omega_{0}}^{'}(h),\delta_{FS,6,\Omega_{0}}^{'}(h)\right\}\\
 \\
\mathbf{1}\left ( \delta_{FS,3}(\Omega_{0})<\min\left\{\delta_{FS,1}(\Omega_{0}),\delta_{FS,2}(\Omega_{0}) \right\} \right )\delta_{FS,3,\Omega_{0}}^{'}(h)\\
+\mathbf{1}\left ( \delta_{FS,2}(\Omega_{0})<\min\left\{\delta_{FS,1}(\Omega_{0}),\delta_{FS,3}(\Omega_{0}) \right\} \right )\delta_{FS,2,\Omega_{0}}^{'}(h) \\
+\mathbf{1}\left ( \delta_{FS,3}(\Omega_{0})<\min\left\{\delta_{FS,1}(\Omega_{0}),\delta_{FS,2}(\Omega_{0}) \right\} \right )\delta_{FS,3,\Omega_{0}}^{'}(h)\\
+\mathbf{1}\left ( \delta_{FS,3}(\Omega_{0})=\delta_{FS,2}(\Omega_{0})<\delta_{FS,1}(\Omega_{0})  \right )\max\left\{ \delta_{FS,3,\Omega_{0}}^{'}(h),\delta_{FS,2,\Omega_{0}}^{'}(h)\right\}\\
+\mathbf{1}\left ( \delta_{FS,3}(\Omega_{0})=\delta_{FS,1}(\Omega_{0})<\delta_{FS,2}(\Omega_{0})  \right )\max\left\{ \delta_{FS,3,\Omega_{0}}^{'}(h),\delta_{FS,1,\Omega_{0}}^{'}(h)\right\}\\
+\mathbf{1}\left ( \delta_{FS,1}(\Omega_{0})=\delta_{FS,2}(\Omega_{0})>\delta_{FS,3}(\Omega_{0})  \right )\max\left\{ \delta_{FS,1,\Omega_{0}}^{'}(h),\delta_{FS,2,\Omega_{0}}^{'}(h)\right\}\\
+\mathbf{1}\left ( \delta_{FS,3}(\Omega_{0})=\delta_{FS,2}(\Omega_{0})=\delta_{FS,1}(\Omega_{0})  \right )\max\left\{ \delta_{FS,3,\Omega_{0}}^{'}(h),\delta_{FS,2,\Omega_{0}}^{'}(h),\delta_{FS,1,\Omega_{0}}^{'}(h)\right\}\\
\end{pmatrix}
\end{equation*}

It follows from the Delta Method that

\begin{equation*}
    \left [ \sqrt{N}\left ( \psi_{FS}(\widehat{\Omega})-\psi_{FS}(\Omega_{0}) \right ) \right ]\overset{d}{\rightarrow}\left [ \psi_{FS,\Omega_{0}}^{'}(\textbf{Z}_{j}) \right ](z,x)\equiv \textbf{Z}_{j,FS}(z,x)
\end{equation*}

For fixed $y$ and $c$ define the mapping

\begin{align*}
        &\psi_{RF}:l^{\infty}\left ( \left\{ 0,1\right\}^{2} \times\mathcal{S}(X)\right )\times l^{\infty}\left ( \left\{ 0,1\right\}^{2}\times\mathcal{S}(X) \right )\times l^{\infty}\left ( \mathcal{S}(X) \right )\\&\rightarrow l^{\infty}\left ( \left\{ 0,1\right\}\times\mathcal{S}(X),\mathbb{R}^{2} \right )
    \end{align*}

    by

    \begin{equation*}
        \left [ \psi_{RF}(\Omega) \right ](z,x)=\begin{pmatrix}
\min\left\{ \delta_{RF,4}(\Omega),\delta_{RF,5}(\Omega),\delta_{RF,6}(\Omega) \right\} \\\max\left\{ \delta_{RF,1}(\Omega),\delta_{RF,2}(\Omega),\delta_{RF,3}(\Omega)\right\}
\end{pmatrix}
    \end{equation*}

    where 

\begin{equation*}
        \left [ \delta_{RF,1}(\Omega) \right ](z,x)=\frac{\Omega^{(1)}(y,z,x)\Omega^{(3)}(z,x)}{\Omega^{(3)}(z,x)+c}
    \end{equation*}

    which has Hadamard directional derivative equal to

    \begin{equation*}
        \left [ \delta^{'}_{RF,1,\Omega}(h) \right ](z,x)=\frac{\Omega^{(1)}(y,z,x)h^{(3)}(z,x)+h^{(1)}(y,z,x)\Omega^{(3)}(z,x)}{\Omega^{(3)}(z,x)+c}-\frac{\Omega^{(1)}(y,z,x)\Omega^{(3)}(z,x)h^{(3)}(z,x)}{\left ( \Omega^{(3)}(z,x)+c \right )^{2}}
    \end{equation*}

    \begin{equation*}
        \left [ \delta_{RF,2}(\Omega) \right ](z,x)=\frac{\Omega^{(1)}(y,z,x)\Omega^{(3)}(z,x)-2c}{\Omega^{(3)}(z,x)-c}
    \end{equation*}

    which has Hadamard directional derivative equal to

    \begin{equation*}
        \left [ \delta^{'}_{RF,2,\Omega}(h) \right ](z,x)=\frac{\Omega^{(1)}(y,z,x)h^{(3)}(z,x)+h^{(1)}(y,z,x)\Omega^{(3)}(z,x)}{\Omega^{(3)}(z,x)-c}-\frac{\left ( \Omega^{(1)}(y,z,x)\Omega^{(3)}(z,x)-2c \right )h^{(3)}(z,x)}{\left ( \Omega^{(3)}(z,x)-c \right )^{2}},
    \end{equation*}

    \begin{equation*}
        \left [ \delta_{RF,3}(\Omega) \right ](z,x)=\Omega^{(1)}(y,z,x)\Omega^{(3)}(z,x)
    \end{equation*}

    which has Hadamard directional derivative equal to

    \begin{equation*}
        \left [ \delta^{'}_{RF,3,\Omega}(h) \right ](z,x)=h^{(1)}(y,z,x)\Omega^{(3)}(z,x)+\Omega^{(1)}(y,z,x)h^{(3)}(z,x),
    \end{equation*}

    \begin{equation*}
        \left [ \delta_{RF,4}(\Omega) \right ](z,x)=\frac{\Omega^{(1)}(y,z,x)\Omega^{(3)}(z,x)}{\Omega^{(3)}(z,x)-c}
    \end{equation*}

    which has Hadamard directional derivative equal to

    \begin{equation*}
        \left [ \delta^{'}_{RF,4,\Omega}(h) \right ](z,x)=\frac{\Omega^{(1)}(y,z,x)h^{(3)}(z,x)+h^{(1)}(y,z,x)\Omega^{(3)}(z,x)}{\Omega^{(3)}(z,x)-c}-\frac{\Omega^{(1)}(y,z,x)\Omega^{(3)}(z,x)h^{(3)}(z,x)}{\left ( \Omega^{(3)}(z,x)-c \right )^{2}}
    \end{equation*},

\begin{equation*}
    \left [ \delta_{RF,5}(\Omega) \right ](z,x)=\frac{\Omega^{(1)}(y,z,x)\Omega^{(2)}(z,x)+2c}{\Omega^{(2)}(z,x)+c},
\end{equation*}

which has Hadamard derivative equal to

\begin{equation*}
    \left [ \delta_{RF,5,\Omega}^{'}(h) \right ](z,x)=\frac{\Omega^{(1)}(y,z,x)h^{(3)}(z,x)+h^{(1)}(y,z,x)\Omega^{(3)}(z,x)}{\Omega^{(3)}(z,x)+c}-\frac{\left ( \Omega^{(1)}(y,z,x)\Omega^{(3)}(z,x)+2c \right )h^{(3)}(z,x)}{\left ( \Omega^{(3)}(z,x)+c \right )^{2}},
\end{equation*}

and

\begin{equation*}
    \left [ \delta_{RF,6,\Omega}^{'}(h) \right ](z,x)=h^{(1)}(y,z,x)\Omega^{(3)}(z,x)+h^{(3)}(z,x)(\Omega^{(1)}(y,z,x)-1)
\end{equation*}

Hence, the Hadamard directional derivative of the map $\psi_{FS}$ evaluated at $\Omega_{0}$ is 

\begin{equation*}
    \psi_{RF,\Omega_{0}}^{'}(h)=\begin{pmatrix}
\mathbf{1}\left ( \delta_{RF,6}(\Omega_{0})>\max\left\{\delta_{RF,4}(\Omega_{0}),\delta_{RF,5}(\Omega_{0}) \right\} \right )\delta_{RF,6,\Omega_{0}}^{'}(h)\\
+\mathbf{1}\left ( \delta_{RF,5}(\Omega_{0})>\max\left\{\delta_{RF,4}(\Omega_{0}),\delta_{RF,6}(\Omega_{0}) \right\} \right )\delta_{RF,5,\Omega_{0}}^{'}(h) \\
+\mathbf{1}\left ( \delta_{RF,4}(\Omega_{0})>\max\left\{\delta_{RF,5}(\Omega_{0}),\delta_{RF,6}(\Omega_{0}) \right\} \right )\delta_{RF,4,\Omega_{0}}^{'}(h)\\
+\mathbf{1}\left ( \delta_{RF,6}(\Omega_{0})=\delta_{RF,5}(\Omega_{0})>\delta_{RF,4}(\Omega_{0})  \right )\min\left\{ \delta_{RF,6,\Omega_{0}}^{'}(h),\delta_{RF,5,\Omega_{0}}^{'}(h)\right\}\\
+\mathbf{1}\left ( \delta_{RF,6}(\Omega_{0})=\delta_{RF,4}(\Omega_{0})>\delta_{RF,5}(\Omega_{0})  \right )\min\left\{ \delta_{RF,6,\Omega_{0}}^{'}(h),\delta_{RF,4,\Omega_{0}}^{'}(h)\right\}\\
+\mathbf{1}\left ( \delta_{RF,4}(\Omega_{0})=\delta_{RF,5}(\Omega_{0})>\delta_{RF,6}(\Omega_{0})  \right )\min\left\{ \delta_{RF,4,\Omega_{0}}^{'}(h),\delta_{RF,5,\Omega_{0}}^{'}(h)\right\}\\
+\mathbf{1}\left ( \delta_{RF,6}(\Omega_{0})=\delta_{RF,5}(\Omega_{0})=\delta_{RF,4}(\Omega_{0})  \right )\min\left\{ \delta_{RF,6,\Omega_{0}}^{'}(h),\delta_{RF,5,\Omega_{0}}^{'}(h),\delta_{RF,6,\Omega_{0}}^{'}(h)\right\}\\
 \\
\mathbf{1}\left ( \delta_{RF,3}(\Omega_{0})<\min\left\{\delta_{RF,1}(\Omega_{0}),\delta_{RF,2}(\Omega_{0}) \right\} \right )\delta_{RF,3,\Omega_{0}}^{'}(h)\\
+\mathbf{1}\left ( \delta_{RF,2}(\Omega_{0})<\min\left\{\delta_{RF,1}(\Omega_{0}),\delta_{RF,3}(\Omega_{0}) \right\} \right )\delta_{RF,2,\Omega_{0}}^{'}(h) \\
+\mathbf{1}\left ( \delta_{RF,3}(\Omega_{0})<\min\left\{\delta_{RF,1}(\Omega_{0}),\delta_{RF,2}(\Omega_{0}) \right\} \right )\delta_{RF,3,\Omega_{0}}^{'}(h)\\
+\mathbf{1}\left ( \delta_{RF,3}(\Omega_{0})=\delta_{RF,2}(\Omega_{0})<\delta_{RF,1}(\Omega_{0})  \right )\max\left\{ \delta_{RF,3,\Omega_{0}}^{'}(h),\delta_{RF,2,\Omega_{0}}^{'}(h)\right\}\\
+\mathbf{1}\left ( \delta_{RF,3}(\Omega_{0})=\delta_{RF,1}(\Omega_{0})<\delta_{RF,2}(\Omega_{0})  \right )\max\left\{ \delta_{RF,3,\Omega_{0}}^{'}(h),\delta_{RF,1,\Omega_{0}}^{'}(h)\right\}\\
+\mathbf{1}\left ( \delta_{RF,1}(\Omega_{0})=\delta_{RF,2}(\Omega_{0})>\delta_{RF,3}(\Omega_{0})  \right )\max\left\{ \delta_{RF,1,\Omega_{0}}^{'}(h),\delta_{RF,2,\Omega_{0}}^{'}(h)\right\}\\
+\mathbf{1}\left ( \delta_{RF,3}(\Omega_{0})=\delta_{RF,2}(\Omega_{0})=\delta_{RF,1}(\Omega_{0})  \right )\max\left\{ \delta_{RF,3,\Omega_{0}}^{'}(h),\delta_{RF,2,\Omega_{0}}^{'}(h),\delta_{RF,1,\Omega_{0}}^{'}(h)\right\}\\
\end{pmatrix}
\end{equation*}
   
It follows from the Delta Method that

\begin{equation*}
    \left [ \sqrt{N}\left ( \psi_{RF}(\widehat{\Omega})-\psi_{RF}(\Omega_{0}) \right ) \right ]\overset{d}{\rightarrow}\left [ \psi_{RF,\Omega_{0}}^{'}(\textbf{Z}_{j}) \right ](z,x)\equiv \textbf{Z}_{j,RF}(z,x)
\end{equation*}

which concludes the proof.

\end{proof}

\end{document}